\documentclass[times,preprint]{elsarticle}
\usepackage{amsmath}
\usepackage{microtype}
\usepackage{amsthm}
\usepackage{amsfonts}
\usepackage{amssymb}
\usepackage{graphicx,psfrag,epsf}
\graphicspath{{figures/}}
\usepackage{natbib}
\usepackage{booktabs}
\usepackage{url}
\usepackage{xcolor}
\usepackage[linesnumbered, ruled, vlined]{algorithm2e}
\usepackage{placeins} %

\theoremstyle{plain}
\newtheorem{theorem}{Theorem}
\newtheorem{lemma}{Lemma}
\newtheorem{corollary}{Corollary}
\theoremstyle{remark}
\newtheorem{remark}{Remark}

\DeclareMathOperator{\Var}{Var}

\DeclareMathOperator{\diag}{diag}

\newcommand{\T}{\mathsf{T}}
\newcommand{\bs}[1]{\boldsymbol{#1}}

\newcommand{\muZ}{\mu_{Z|XY}}
\newcommand{\sigZ}{\sigma_{Z|XY}}

\newif\ifDRAFT \DRAFTfalse   %
\ifDRAFT
  \usepackage[textsize=footnotesize,colorinlistoftodos]{todonotes}
\else
  \usepackage[disable]{todonotes}
\fi

\newcommand{\CAPNOTE}[1]{\ifDRAFT~\textcolor{blue!70!black}{\textbf{[NOTE: #1]}}\fi}

\newcommand{\maeProposed}{3.1 \times 10^{-5}}

\newcommand{\maeDelta}{0.1717}

\newcommand{\maxDeltaZeroMean}{0.45}
\newcommand{\maxDeltaNonZeroMean}{0.2030}
\newcommand{\deltaNZMReduction}{reduced that maximum by more than half}

\newcommand{\timeEvalLaptopLo}{0.28}
\newcommand{\timeEvalLaptopHi}{0.35}
\newcommand{\timeCILaptopLo}{8}
\newcommand{\timeCILaptopHi}{15}
\newcommand{\nEvalCILaptopLo}{24}
\newcommand{\nEvalCILaptopHi}{29}
\newcommand{\timeBootstrapEst}{90--240}

\newcommand{\covCells}{42}
\newcommand{\covB}{2000}
\newcommand{\covMCSE}{0.005}
\newcommand{\covTImuZeroLo}{0.990}
\newcommand{\covTImuZeroHi}{1.000}
\newcommand{\covElLo}{0.993}
\newcommand{\covElHi}{1.000}
\newcommand{\covElMuZeroLo}{0.999}
\newcommand{\covElMuZeroHi}{1.000}
\newcommand{\covCCWCmax}{0.0005}
\newcommand{\covTIminAll}{0.942}
\newcommand{\covTIminAllCPlo}{0.930}
\newcommand{\covTIminAllCPhi}{0.951}
\newcommand{\covTInbLo}{0.943}
\newcommand{\covTInbHi}{0.955}
\newcommand{\covDeltaNBfifty}{0.857--0.860}

\newcommand{\covDeltaNBthousand}{0.947}

\newcommand{\covMCTIcovDiff}{0.003}
\newcommand{\covMCTIwidthPct}{0.2}
\newcommand{\covMCTIscorePct}{0.3}
\newcommand{\covFailMaxPct}{0.05}
\newcommand{\covCoreHours}{295}

\usepackage{hyperref}
\hypersetup{colorlinks=true, citecolor=blue!60!black,
            linkcolor=blue!50!black, urlcolor=blue!60!black}

\begin{document}
\setlength{\emergencystretch}{3em}  %
\hbadness=3600  %

\begin{frontmatter}
    \title{Dimension Reduction for the Trivariate Normal Product Distribution: An Efficient Numerical Algorithm via Conditional Expectation}

    \author[unm]{Davood Tofighi\corref{cor1}}
    \ead{dtofighi@unm.edu}
    \address[unm]{Department of Mathematics and Statistics, University of New Mexico, Albuquerque, NM 87131, USA}
    \cortext[cor1]{Corresponding author}

    \begin{abstract}
        Statistical inference in sequential mediation models requires evaluating the cumulative distribution function (CDF) of the product of three normal coefficient estimators, an operation that entails integrating over a non-convex region with boundary $xyz=v$. While the Delta method provides an instantaneous first-order approximation, its asymptotic variance collapses whenever two or more path coefficients approach zero, producing severe undercoverage near parameter boundaries. Nonparametric bootstrapping avoids gradient collapse but incurs an $\mathcal{O}(B \times N)$ computational cost that becomes burdensome in large-scale simulation studies or iterative power analyses. We propose a model-based dimension-reduction algorithm that integrates out the third variable analytically under the trivariate Gaussian distribution for arbitrary mean vectors and positive-definite covariance matrices. Partitioning the $xy$-plane into quadrants isolates the sign change at the coordinate axes, reducing the problem to an adaptive two-dimensional quadrature whose fixed-grid cost is $\mathcal{O}(n^2)$ in place of $\mathcal{O}(n^3)$. In simulation benchmarks across six parameter regimes, the algorithm achieves a mean absolute error of $\maeProposed$ relative to a $10^8$-sample Monte Carlo reference, evaluating the distribution in under one second and yielding plug-in quantile confidence intervals whose empirical coverage was near or above nominal.
    \end{abstract}

    \begin{keyword}
        Numerical integration \sep Multivariate normal distribution \sep Product distribution \sep Dimension reduction \sep Sequential mediation \sep Adaptive quadrature
    \end{keyword}
\end{frontmatter}

\ifDRAFT
    \listoftodos[Pending items \& TODOs (draft only --- set \texttt{\textbackslash DRAFTfalse} to remove)]
\fi

\section{Introduction} \label{sec:intro}

Causal mediation analysis investigates the intermediate pathways through which an exposure transmits its effect to an outcome, decomposing the total causal effect into direct and indirect components \citep{robins1992, pearl2001, mackinnon2008, vanderweele2014c}. In a single-mediator model, such as the JOBS II intervention where assignment to a job-skills workshop ($A$) reduces depressive symptoms ($Y$) by increasing job-search self-efficacy ($M_1$) \citep{vinokur1997}, the indirect effect under linear additive structural models equals the product of two path coefficients \citep{mackinnon2008}.\footnote{For simplicity, we omit baseline covariates; when present, they render each coefficient a partial (covariate-adjusted) effect without altering the underlying product structure.} When mediation operates through a sequential chain of two intermediate variables ($A \rightarrow M_1 \rightarrow M_2 \rightarrow Y$), the workshop first elevates self-efficacy ($M_1$), self-efficacy subsequently enhances job-search behavior ($M_2$), and job-search behavior ultimately reduces depressive symptoms ($Y$). In this two-mediator chain, the indirect effect estimand is the three-path product $\theta = \beta_1 \beta_2 \beta_3$, estimated by the product of the corresponding coefficient estimators $\hat\theta = \hat\beta_1 \hat\beta_2 \hat\beta_3$. Under standard regularity conditions, ordinary least-squares and maximum-likelihood estimators are asymptotically multivariate normal. Statistical inference regarding $\theta$ (including hypothesis tests, confidence intervals, and power calculations) therefore requires the sampling distribution of the product of three correlated normal random variables.

For the product of two normal random variables, the distributional theory is well understood. \citet{pearson1929} derived the distribution of the sample product-moment coefficient, while \citet{wishart1932} obtained the density of the product of two zero-mean normal variables in terms of the modified Bessel function of the second kind. \citet{craig1936} extended this formulation to non-zero means using a Laurent series expansion. Subsequent work characterized products of independent normal variables \citep{springer1966} as well as the product moments of a general $K$-fold multivariate normal distribution \citep{goodman1960, goodman1962, bohrnstedt1969}. Because low-order moments do not determine tail quantiles, attention shifted to exact distribution functions. \citet{meeker1994} developed numerical quadrature routines for the bivariate product distribution. \citet{cui2016} later derived an exact analytic density for correlated normal variables with arbitrary means as an infinite series of modified Bessel functions, and \citet{gaunt2025} developed asymptotic approximations for tail quantiles. For single-mediator models, confidence intervals based on the distribution of the product of two normal variables are standard practice \citep{mackinnon2008, tofighi2011}.

For products of three or more correlated normal variables, no closed-form or series representation free of residual integration is available. \citet{ogasawara2025} expressed the density of a product of $K$ correlated zero-mean normal variables as a $(K-1)$-fold integral, as well as a $(K-2)$-fold integral over an infinite double Bessel series. At $K = 3$, both representations retain residual integration, leaving an analytic expression an open problem. Consequently, for the sequential indirect effect estimator $\hat\theta = \hat\beta_1 \hat\beta_2 \hat\beta_3$ in a two-mediator chain \citep{vanderweele2014c, vanderweele2015}, the cumulative distribution function (CDF) has no known closed form free of residual integration. The central obstacle is geometric: the product threshold defines a non-convex integration region with boundary $xyz=v$. Solving this boundary for $z$ gives the sign-dependent expression $v/(xy)$, which is undefined when $x=0$ or $y=0$ and complicates general-purpose multidimensional cubature \citep{berntsen1991}.

In the absence of an exact distribution function, applied researchers rely on approximations or resampling methods. The Delta method linearizes the product, but its first-order variance collapses to zero whenever any two path coefficients approach zero (Lemma~\ref{lem:delta}, \ref{app:delta}), producing severe power loss and boundary undercoverage in small-to-moderate samples \citep{lomnicki1967, mackinnon2002}. In practice, the plug-in Wald statistic is bounded by the smallest individually standardized coefficient, preventing rejection of the null hypothesis unless every path is individually significant (Remark~\ref{rem:plugin}, \ref{app:delta}; Section~\ref{subsec:coverage}). Moment-matching approximations \citep{zheng2012} calibrate only low-order moments, failing in the extreme tails that govern $p$-values and coverage. Nonparametric bootstrapping \citep{preacher2008} avoids distributional assumptions but relies on asymptotic justification \citep{bickel1981}, degrades near parameter boundaries \citep{andrews2000}, and incurs an $\mathcal{O}(B \times N)$ computational cost that becomes burdensome in large-scale simulation studies or iterative power analyses.

We propose a dimension-reduction algorithm for computing the CDF of the sequential indirect effect estimator $\hat\theta = \hat\beta_1 \hat\beta_2 \hat\beta_3$ under joint normality. Rather than introducing a distributional approximation or relying on Monte Carlo simulation, we integrate one variable analytically by conditioning on the remaining two. This conditional Gaussian structure expresses the innermost integral as a standard normal tail probability, leaving a bivariate integrand that is continuous on each open quadrant. For comparable fixed tensor-product quadrature rules, this reduction changes the grid cost from $\mathcal{O}(n^3)$ to $\mathcal{O}(n^2)$; adaptive quadrature cost is instead measured by the observed number of integrand evaluations at a stated tolerance. The implemented algorithm achieves a mean absolute error of $\maeProposed$ relative to a $10^8$-sample Monte Carlo reference benchmark across diverse covariance structures (Section~\ref{sec:simulation}) with evaluation times of \timeEvalLaptopLo--\timeEvalLaptopHi{} seconds on a standard laptop (Section~\ref{subsec:software}).

\subsection{Justification and Rationale}

The central computational difficulty is the geometry of the integration region. Evaluating $P(XYZ < v)$ requires integrating the trivariate Gaussian density over a non-convex region whose boundary is a level set, that is, the set on which the product equals the threshold; away from $x=0$ and $y=0$, this boundary is represented by $z = v/(xy)$. Unlike the rectangular integration domains encountered in standard multivariate normal probability routines \citep{drezner1994, genz2009}, the integration limit for one variable depends on the values and signs of the other two. A direct three-dimensional tensor-product quadrature rule using $n$ nodes per coordinate requires $\mathcal{O}(n^3)$ integrand evaluations. For adaptive three-dimensional quadrature, computational work is instead reported as the observed number of integrand evaluations at a stated tolerance because it depends on the refinement strategy and integrand geometry.

Conditioning on $(X, Y)$ yields three computational consequences:
\begin{enumerate}
    \item \textbf{Dimension reduction:} Analytically marginalizing out $Z$ conditional on $(X, Y)$ leaves a two-dimensional residual integration and, for comparable fixed tensor-product quadrature rules, reduces grid cost from $\mathcal{O}(n^3)$ to $\mathcal{O}(n^2)$.
    \item \textbf{Evaluation of the inner probability:} The conditional distribution $Z \mid X, Y \sim \mathcal{N}(\muZ, \sigZ^2)$ allows the inner conditional probability to be expressed in terms of the univariate normal CDF $\Phi(\cdot)$; the sign of the product of the conditioning values determines whether this probability is a lower- or upper-tail probability.
    \item \textbf{Model-implied distributional target:} Unlike the Delta method, which linearizes the product, the proposed formulation introduces no structural approximation to the product distribution under the assumed joint trivariate normal model. The residual integration is evaluated numerically at a user-specified tolerance, although adaptive error control does not guarantee that the requested tolerance is achieved (Section~\ref{subsubsec:cubature}).
\end{enumerate}

The proposed algorithm computes the CDF directly for arbitrary mean vectors $\bs{\mu}$ and positive-definite covariance matrices $\bs{\Sigma}$, complementing the correlated zero-mean setting where existing results deliver only the density or characteristic function \citep{ogasawara2025}.

\section{Mathematical Preliminaries}
\label{sec:back}

Let $\bs{X} = [X, Y, Z]^{\T}$ be a random vector following a trivariate normal distribution with mean vector $\bs{\mu} = [\mu_X, \mu_Y, \mu_Z]^{\T}$ and positive-definite covariance matrix $\bs{\Sigma} \succ 0$:
\begin{equation}
    \bs{X} \sim \mathcal{N}_3(\bs{\mu}, \bs{\Sigma}), \quad \text{where } \bs{\Sigma} = \begin{bmatrix}
        \sigma_X^2  & \sigma_{XY} & \sigma_{XZ} \\
        \sigma_{XY} & \sigma_Y^2  & \sigma_{YZ} \\
        \sigma_{XZ} & \sigma_{YZ} & \sigma_Z^2
    \end{bmatrix} .
\end{equation}

Our dimension-reduction approach builds on the conditional distribution of $Z$ given $(X,Y)=(x,y)$. Under joint normality, this conditional distribution is univariate normal \citep[Theorem~2.5.1]{anderson2003}:
\begin{equation}
    \label{eq:conditional_dist}
    Z \mid X=x, Y=y \sim \mathcal{N}(\muZ, \sigZ^2) .
\end{equation}
The conditional mean and conditional variance are given by:
\begin{align}
    \muZ(x,y) & = \mu_Z + \bs{\Sigma}_{Z,XY} \bs{\Sigma}_{XY}^{-1} \begin{bmatrix} x - \mu_X \\ y - \mu_Y \end{bmatrix} \label{eq:cond_mean} , \\
    \sigZ^2   & = \sigma_Z^2 - \bs{\Sigma}_{Z,XY} \bs{\Sigma}_{XY}^{-1} \bs{\Sigma}_{Z,XY}^{\T} \label{eq:cond_var} ,
\end{align}
where $\bs{\Sigma}_{XY}$ is the leading $2 \times 2$ principal submatrix of $\bs{\Sigma}$, and $\bs{\Sigma}_{Z,XY} = [\sigma_{XZ}, \sigma_{YZ}]$ is the cross-covariance block. Because $\bs{\Sigma} \succ 0$, the principal submatrix $\bs{\Sigma}_{XY}$ is strictly positive definite and invertible. The conditional variance $\sigZ^2$ is the Schur complement of $\bs{\Sigma}_{XY}$ in $\bs{\Sigma}$, implying that $\sigZ^2 > 0$. In particular, $\sigZ^2$ does not depend on the conditioning values $(x,y)$, whereas $\muZ(x,y)$ is affine in $(x,y)$. The conditional distribution is therefore completely determined by the mean and variance above. Its CDF gives the inner integral over $z$, leaving a two-dimensional integral over $(x,y)$.

\section{The Model-Based Dimension Reduction Method}
\label{sec:prop}

We derive the CDF of the product $V = XYZ$. In sequential mediation, these variables represent the three estimated path coefficients $(\hat\beta_1, \hat\beta_2, \hat\beta_3)$; $\bs{\Sigma}$ denotes their asymptotic covariance matrix, and $V$ corresponds to the estimated indirect effect $\hat\theta$. Accordingly, $X$, $Y$, and $Z$ represent coefficient estimators rather than the exposure and outcome variables of Section~\ref{sec:intro}; \ref{app:delta} details this mapping. Formally, the CDF is the probability assigned to the region $\mathcal{A}(v) = \{(x,y,z) \in \mathbb{R}^3 : xyz < v\}$ under the joint distribution of $(X, Y, Z)$.

Direct evaluation requires three-dimensional quadrature over the non-convex region $\mathcal{A}(v)$. Its boundary is the product level set $L_v = \{(x,y,z) \in \mathbb{R}^3 : xyz = v\}$. When $xy \neq 0$, this boundary can be written as $z = v/(xy)$; this representation is undefined when $x=0$ or $y=0$. The constraint on $z$ depends on the sign of $xy$:
\begin{equation}
    \begin{cases}
        z < \frac{v}{xy} & \text{if } xy > 0 ,\\
        z > \frac{v}{xy} & \text{if } xy < 0 .
    \end{cases}
\end{equation}
The case $xy=0$ occurs on the two coordinate axes of the $xy$-plane. Because $\bs{\Sigma}_{XY}$ is positive definite, $(X,Y)$ has a bivariate normal density, and these axes have probability zero.
A direct fixed tensor-product quadrature rule over this domain has cubic cost in the number of nodes per coordinate, whereas direct adaptive quadrature must account for the non-convex domain and the sign-dependent integration limits. We partition the $xy$-plane into quadrants and integrate out $Z$ analytically.

\subsection{Analytic Marginalization via Conditional Expectation}

By iterated expectation, the CDF is the expectation, over $(X,Y)$, of the conditional probability $P(XYZ < v \mid X,Y)$. Factoring the joint density as $f_{X,Y,Z}(x,y,z) = f_{Z|X,Y}(z \mid x,y)f_{X,Y}(x,y)$ expresses this conditional probability as the inner integral over $z$. The integration limit depends on the sign of $xy$ (Quadrants I and III versus Quadrants II and IV). Because $Z \mid X,Y$ is univariate normal by Equation~(\ref{eq:conditional_dist}), the conditional probability is a lower- or upper-tail probability and can be written in terms of $\Phi(\cdot)$.

\subsection{The Main Computational Formula}

Combining the two conditional tail probabilities reduces the remaining numerical integration from three dimensions to two and lowers the cost of a fixed tensor-product quadrature rule from $\mathcal{O}(n^3)$ to $\mathcal{O}(n^2)$.

\begin{theorem}[Distribution of the Product of Three Normal Variables]
    \label{thm:main}
    Let $\bs{X} = [X, Y, Z]^{\T} \sim \mathcal{N}_3(\bs{\mu}, \bs{\Sigma})$ with $\bs{\Sigma} \succ 0$, as defined in Section~\ref{sec:back}. The CDF of $V = XYZ$ is given by:
    \begin{equation}
        \label{eq:cdf_final}
        \boxed{
            F_V(v) = \int_{-\infty}^{\infty} \int_{-\infty}^{\infty} \Phi\left(\text{sgn}(xy) \cdot \frac{v/(xy) - \muZ(x,y)}{\sigZ}\right) f_{X,Y}(x,y) \, dy \, dx
        }
    \end{equation}
    where $\text{sgn}(\cdot)$ is the signum function, $\muZ(x,y)$ is the conditional mean function in Eq.~(\ref{eq:cond_mean}), $\sigZ$ is the conditional standard deviation in Eq.~(\ref{eq:cond_var}), and $f_{X,Y}(x,y)$ is the bivariate normal PDF of $(X,Y)$.
\end{theorem}

The formula in Equation~(\ref{eq:cdf_final}) contains $v/(xy)$ and is therefore undefined when $x=0$ or $y=0$. These two coordinate axes have two-dimensional Lebesgue measure zero. Because $(X,Y)$ has a bivariate normal density, the probability that $(X,Y)$ lies on either axis is zero. The integrand may therefore take any value on the axes, such as zero, without changing the integral.

\begin{proof}
    We begin with the definition of the CDF:
    \begin{equation}
        F_V(v) = P(XYZ < v) = \iiint_{xyz < v} f_{X,Y,Z}(x,y,z) \, dz \, dy \, dx .
    \end{equation}
    The product level set
    \[
        L_v = \{(x,y,z) \in \mathbb{R}^3 : xyz = v\}
    \]
    has three-dimensional Lebesgue measure zero. Because $(X,Y,Z)$ has a joint density, $P(V=v)=0$; thus $P(V<v)=P(V\leq v)$, and the strict inequality above agrees with the standard CDF definition $F_V(v) = P(V \leq v)$. By conditional density factorization,
    \begin{equation}
        f_{X,Y,Z}(x,y,z) = f_{Z|X,Y}(z \mid x,y) \cdot f_{X,Y}(x,y) .
    \end{equation}
    Substituting this factorization into the CDF and applying Tonelli's theorem for non-negative measurable functions \citep[Theorem~5.28]{axler2020} yields the iterated integral:
    \begin{equation}
        \begin{aligned}
            F_V(v)
            &= \int_{-\infty}^{\infty} \int_{-\infty}^{\infty}
               \left[ \int_{xyz < v} f_{Z|X,Y}(z \mid x,y) \, dz \right] \\
            &\qquad{} \times f_{X,Y}(x,y) \, dy \, dx .
        \end{aligned}
    \end{equation}

    Determining the integration limits for $z$ conditional on $(X, Y) = (x, y)$ splits on the sign of $xy$. When $xy > 0$, the limits for $z$ are $(-\infty, v/(xy))$:
    \begin{equation}
        \begin{aligned}
            &\int_{-\infty}^{v/(xy)} f_{Z|X,Y}(z \mid x,y) \, dz \\
            &\quad = P\left(Z \leq \frac{v}{xy} \;\middle|\; X=x, Y=y\right) \\
            &\quad = \Phi\left(\frac{v/(xy) - \muZ(x,y)}{\sigZ}\right) .
        \end{aligned}
    \end{equation}
    When $xy < 0$, the integration limits for $z$ become $(v/(xy), \infty)$, and
    \begin{equation}
        \begin{aligned}
            &\int_{v/(xy)}^{\infty} f_{Z|X,Y}(z \mid x,y) \, dz \\
            &\quad = 1 - \Phi\left(\frac{v/(xy) - \muZ(x,y)}{\sigZ}\right) \\
            &\quad = \Phi\left(-\frac{v/(xy) - \muZ(x,y)}{\sigZ}\right),
        \end{aligned}
    \end{equation}
    where the second equality applies standard normal symmetry $1 - \Phi(t) = \Phi(-t)$.

    Unifying both cases via the signum function $\text{sgn}(xy)$ yields:
    \begin{equation}
        \int_{xyz < v} f_{Z|X,Y}(z \mid x,y) \, dz = \Phi\left(\text{sgn}(xy) \cdot \frac{v/(xy) - \muZ(x,y)}{\sigZ}\right)
    \end{equation}
    Substituting this result into the iterated integral yields Equation~(\ref{eq:cdf_final}). The set $\{\,(x,y) : xy = 0\,\}$ is the union of the two coordinate axes. It has two-dimensional Lebesgue measure zero and, because $(X,Y)$ has a bivariate normal density, has probability zero. Theorem~\ref{thm:convergence} and \ref{app:convergence} supply the formal convergence proof.
\end{proof}

This reduction yields a two-dimensional integral over $\mathbb{R}^2$ whose integrand is continuous on each open quadrant, as shown in the proof of the convergence theorem; adaptive quadrature is applied separately to these quadrants.

\begin{remark}[Density Representation and the Singularity at Zero]
    Fix $(x,y)$ with $xy \neq 0$. Conditionally on $(X,Y)=(x,y)$, the product $V=xyZ$ is univariate normal with mean $xy\,\muZ(x,y)$ and standard deviation $|xy|\,\sigZ$. Its conditional density at $t \in \mathbb{R}$ is
    \begin{equation*}
        k(t \mid x,y)
        = \frac{1}{|xy|\,\sigZ}
          \phi\!\left(\frac{t/(xy)-\muZ(x,y)}{\sigZ}\right),
    \end{equation*}
    where $\phi(\cdot)$ is the standard normal density. The absolute value is the Jacobian for the transformation $z \mapsto xy z$; when $xy<0$, the apparent sign change inside $\phi$ is immaterial because $\phi$ is even.

    Set $k(t \mid x,y)=0$ when $xy=0$, which does not change any integral because the coordinate axes have probability zero, and define
    \begin{equation*}
        g(t) = \iint_{\mathbb{R}^2}
        k(t \mid x,y) f_{X,Y}(x,y) \, dy \, dx .
    \end{equation*}
    The map $(t,x,y) \mapsto k(t \mid x,y)f_{X,Y}(x,y)$ is nonnegative and measurable: it is continuous on the open set where $xy\neq0$ and equals zero on the closed set $\{(t,x,y):xy=0\}$. Tonelli's theorem \citep[Theorem~5.28]{axler2020} therefore permits exchanging the order of integration. The inner integral below equals one when $xy\neq0$ and zero on the coordinate axes; because those axes have probability zero,
    \begin{equation*}
        \begin{aligned}
            \int_{-\infty}^{\infty} g(t)\,dt
            &= \iint_{\mathbb{R}^2}
               \left[\int_{-\infty}^{\infty} k(t \mid x,y)\,dt\right]
               f_{X,Y}(x,y) \, dy \, dx \\
            &= 1 .
        \end{aligned}
    \end{equation*}
    For each $(x,y)$ with $xy\neq0$, the integral $\int_{-\infty}^{v}k(t\mid x,y)\,dt$ is the conditional probability $P(xyZ<v\mid X=x,Y=y)$ that appears in Equation~(\ref{eq:cdf_final}). Conditioning therefore gives
    \begin{equation*}
        F_V(v)
        = \iint_{\mathbb{R}^2}
          \left[\int_{-\infty}^{v} k(t\mid x,y)\,dt\right]
          f_{X,Y}(x,y)\,dy\,dx .
    \end{equation*}
    A second application of Tonelli's theorem exchanges the integrals and gives
    \begin{equation*}
        F_V(v) = \int_{-\infty}^{v} g(t)\,dt .
    \end{equation*}
    Thus, $g$ is a density of $V$. Moreover, $F_V$ is absolutely continuous because, for every $a<b$, $F_V(b)-F_V(a)=\int_a^b g(t)\,dt$. The Lebesgue differentiation theorem \citep[Theorem~4.19]{axler2020} then gives $F_V'(v)=g(v)$ for almost every $v$. In particular, a version of the density is
    \begin{equation}
        \label{eq:pdf_final}
        \begin{aligned}
            f_V(v)
            &= \iint_{\mathbb{R}^2}
               \frac{1}{|xy|\,\sigZ}
               \phi\!\left(\frac{v/(xy)-\muZ(x,y)}{\sigZ}\right) \\
            &\qquad{} \times f_{X,Y}(x,y) \, dy \, dx .
        \end{aligned}
    \end{equation}
    The conditional-density construction and Tonelli's theorem therefore yield the density without differentiating Equation~(\ref{eq:cdf_final}) under the integral sign. Direct differentiation would require exchanging $d/dv$ with the double integral over $(x,y)$. A standard sufficient condition for this exchange at a given $v_0$ is the existence of a single integrable function that bounds the differentiated integrand for all $v$ in a neighborhood of $v_0$, including near the coordinate axes.

    At $v=0$, the integrand in Equation~(\ref{eq:pdf_final}) becomes
    \begin{equation*}
        \frac{\phi\big(\muZ(x,y)/\sigZ\big)f_{X,Y}(x,y)}
        {|xy|\sigZ},
    \end{equation*}
    Write this integrand as $H(x,y)/|xy|$, where
    \begin{equation*}
        H(x,y)=\frac{\phi\big(\muZ(x,y)/\sigZ\big)f_{X,Y}(x,y)}{\sigZ}.
    \end{equation*}
    The function $H$ is continuous and strictly positive. Fix $y_0\neq0$ and choose a sufficiently small compact rectangle around $(0,y_0)$ that does not intersect the $x$-axis. On this rectangle, $H(x,y)/|y|$ is bounded below by a positive constant. After integration over the rectangle's $y$-interval, the double integral is therefore bounded below, for some $C>0$, by
    \begin{equation*}
        C\int_{-\epsilon}^{\epsilon}\frac{dx}{|x|}=+\infty .
    \end{equation*}
    The improper integral is infinite because its values over punctured neighborhoods grow without bound, not because the integrand is undefined at $x=0$. Hence, Equation~(\ref{eq:pdf_final}) defines a density version with $f_V(0)=+\infty$. However, there is no point mass at zero: $\{V=0\}=\{X=0\}\cup\{Y=0\}\cup\{Z=0\}$, and each nondegenerate normal marginal assigns probability zero to a singleton. Thus, an infinite density value on the Lebesgue-null set $\{0\}$ is consistent with $\Pr(V=0)=0$ and a continuous CDF.

    Finally, substituting $k(v \mid x,y)=f_{Z|X,Y}\big(v/(xy) \mid x,y\big)/|xy|$ and using $f_{Z|X,Y}f_{X,Y}=f_{X,Y,Z}$ gives
    \begin{equation*}
        f_V(v) = \iint_{\mathbb{R}^2} \frac{1}{|xy|}\,
        f_{X,Y,Z}\!\left(x,y,\frac{v}{xy}\right) \, dy \, dx ,
    \end{equation*}
    which is the product-Jacobian representation of the density \citep[Result~2]{ogasawara2025}.
\end{remark}

\subsection{Computational Algorithm}

We translate the formulation of Theorem~\ref{thm:main} into Algorithm~\ref{alg:cdf_computation}. To ensure numerical stability when correlations approach $\pm 1$, we avoid explicit matrix inversion and compute the conditional-mean regression vector $\bs{b}_{Z|XY}$ via the Cholesky factor of the leading submatrix $\bs{R}_{XY}$ with forward and backward triangular substitutions. Standardizing the variables to unit marginal variances removes scale disparities across variables and enhances numerical conditioning.

\begin{algorithm}[!ht]
    \SetAlgoLined
    \SetFuncSty{textsc}
    \SetAlFnt{\footnotesize}
    \SetDataSty{textrm}
    \SetNlSty{textrm}{}{:}

    \caption{Computation of $F_V(v)$ via Dimension Reduction}
    \label{alg:cdf_computation}

    \KwIn{Mean vector $\bs{\mu} = [\mu_X, \mu_Y, \mu_Z]^{\T}$, covariance matrix $\bs{\Sigma}$, threshold $v$, near-axis cutoff $\epsilon_{\mathrm{axis}} = 10^{-10}$.}
    \KwOut{$F_V(v) = P(XYZ < v)$}

    \tcp{Preprocessing: Standardization}
    \BlankLine
    Extract standard deviations: $\sigma_X \leftarrow \sqrt{\Sigma_{11}}$, $\sigma_Y \leftarrow \sqrt{\Sigma_{22}}$, $\sigma_Z \leftarrow \sqrt{\Sigma_{33}}$\label{ln:pre-first}\;
    Compute correlation matrix: $\bs{R} \leftarrow \diag(\bs{\sigma})^{-1} \bs{\Sigma} \, \diag(\bs{\sigma})^{-1}$\;
    Standardize means: $\bs{\mu}^* \leftarrow [\mu_X/\sigma_X, \mu_Y/\sigma_Y, \mu_Z/\sigma_Z]^{\T}$\;
    Standardize threshold: $v^* \leftarrow v / (\sigma_X \sigma_Y \sigma_Z)$\label{ln:pre-last}\;

    \tcp{Step 1: Compute Conditional Distribution Parameters}
    \BlankLine
    Partition $\bs{R}$ into $\bs{R}_{XY} \leftarrow \bs{R}[1{:}2, 1{:}2]$ and $\bs{R}_{Z,XY} \leftarrow \bs{R}[3, 1{:}2]$\label{ln:cond-first}\;
    Compute Cholesky decomposition: $\bs{L}\bs{L}^{\T} \leftarrow \bs{R}_{XY}$\label{ln:chol-first}\;
    \tcp{Solve $\bs{R}_{XY} \bs{b}_{Z|XY} = \bs{R}_{Z,XY}^{\T}$ via forward/backward substitution}
    $\bs{w} \leftarrow \text{ForwardSolve}(\bs{L}, \bs{R}_{Z,XY}^{\T})$\;
    $\bs{b}_{Z|XY} \leftarrow \text{BackwardSolve}(\bs{L}^{\T}, \bs{w})$\label{ln:chol-last}\;
    Compute conditional variance: $\sigZ^2 \leftarrow 1 - \bs{R}_{Z,XY} \bs{b}_{Z|XY}$; set $\sigZ \leftarrow \sqrt{\sigZ^2}$\label{ln:cond-last}\;

    \tcp{Step 2: Define Integrand Function $g(x,y)$}
    \BlankLine
    \SetKwFunction{FunctG}{$g(x, y)$}
    \SetKwProg{Fn}{Function}{:}{End Function}
    \Fn{\FunctG{x, y}}{\label{ln:fn-first}
        Compute conditional mean: $\muZ(x,y) \leftarrow \mu_Z^* + \bs{b}_{Z|XY}^{\T} [x - \mu_X^*, y - \mu_Y^*]$ \;
        Compute bivariate normal PDF: $f_{X,Y}(x,y) \leftarrow \frac{1}{2\pi\sqrt{|\bs{R}_{XY}|}} \exp\left(-\frac{1}{2}\bs{z}^{\T} \bs{R}_{XY}^{-1} \bs{z}\right)$ \tcp*{where $\bs{z} = [x - \mu_X^*, y - \mu_Y^*]^{\T}$}
        \If{$|xy| < \epsilon_{\mathrm{axis}}$}{\label{ln:axis-first}
            $t \leftarrow \text{sgn}(v^*) \cdot 10^{10}$ \tcp*{Near-axis overflow guard}
        }
        \Else{
            $t \leftarrow \text{sgn}(xy) \cdot \frac{v^*/(xy) - \muZ(x,y)}{\sigZ}$\label{ln:axis-last}\;
        }
        \KwRet{$\Phi(t) \cdot f_{X,Y}(x,y)$}\label{ln:fn-last}\;
    }

    \tcp{Step 3: Numerical Integration (split at the coordinate axes)}
    \BlankLine
    Set integration bounds: $[a_x, b_x] \leftarrow [\mu_X^* - 8, \mu_X^* + 8]$, $[a_y, b_y] \leftarrow [\mu_Y^* - 8, \mu_Y^* + 8]$\label{ln:int-first}\;
    Split the box at $x=0$ and $y=0$ into quadrant cells $\{C_k\}$, clamping split points to the bounds\;
    \tcp*{$\text{sgn}(xy)$ is constant inside each cell}
    Evaluate $I \leftarrow \sum_k \int_{C_k} g(x, y) \, dx \, dy$ by adaptive cubature on each cell\;
    \tcp*{H-adaptive cubature (\texttt{hcubature}) per cell, relative tolerance $10^{-6}$}

    \KwRet{$I$}\label{ln:int-last}

\end{algorithm}
\FloatBarrier

\subsection{Complexity Analysis}
\label{sec:complexity}

The computational complexity of Algorithm~\ref{alg:cdf_computation} divides into three distinct stages:

\textbf{Preprocessing and Parameter Setup (Lines~\ref{ln:pre-first}--\ref{ln:pre-last}):} Computing the correlation matrix $\bs{R}$, scaling the mean vector to $\bs{\mu}^*$, and standardizing the evaluation threshold to $v^*$ requires a fixed number of arithmetic operations, giving $\mathcal{O}(1)$ time complexity for fixed dimension $d=3$.

\textbf{Conditional Distribution Parameters (Lines~\ref{ln:cond-first}--\ref{ln:cond-last}):} The $2 \times 2$ Cholesky factorization $\bs{L}\bs{L}^{\T} = \bs{R}_{XY}$ and subsequent triangular substitutions compute the coefficient vector $\bs{b}_{Z|XY}$ and conditional variance $\sigZ^2$ in $\mathcal{O}(1)$ operations. Because $\bs{b}_{Z|XY}$ and $\sigZ$ are invariant to the integration variables $(x, y)$, these parameters are computed once prior to quadrature.

\textbf{Numerical Quadrature (Lines~\ref{ln:int-first}--\ref{ln:int-last}):} The dominant cost arises from evaluating the two-dimensional integral over the quadrant cells $\{C_k\}$. Each call to the integrand $g(x,y)$ (Lines~\ref{ln:fn-first}--\ref{ln:fn-last}) requires $\mathcal{O}(1)$ operations, comprising one affine transformation for $\muZ(x,y)$, one bivariate Gaussian density evaluation $f_{X,Y}(x,y)$, and one univariate normal CDF evaluation $\Phi(t)$.

For a classical tensor-product grid with $n = \mathcal{O}(h^{-1})$ nodes per coordinate, evaluating the double integral requires $N = \mathcal{O}(n^2)$ integrand evaluations. On smooth functions with convergence order $r$, achieving a target error tolerance $\delta$ requires $N = \mathcal{O}(\delta^{-2/r})$ evaluations \citep[\S\S5.6, 5.8]{davis1984}. Direct three-dimensional quadrature over the non-convex domain requires $\mathcal{O}(n^3)$ grid evaluations. Complexity theory establishes that no $N$-point linear quadrature rule can surpass the intrinsic lower error bound for a fixed smoothness class \citep{bakhvalov1959, davis1984}, and adaptive point selection cannot exceed this optimal rate on convex, balanced function classes \citep[Theorem~0]{novak2016}.

In practice, our implementation uses globally adaptive $h$-subdivision \citep{berntsen1991} based on a degree-7 Genz--Malik cubature rule with an embedded degree-5 error estimator \citep{genzmalik1980}. This approach concentrates function evaluations near steep gradients without constructing a fine uniform mesh over the entire domain. Across the six simulation regimes of Section~\ref{sec:simulation} and thresholds $q \in \{-1, 0, 0.5, 1\}$, evaluation counts at $\delta = 10^{-6}$ ranged from roughly $0.9 \times 10^{4}$ to $5.0 \times 10^{4}$; these counts describe the evaluated cases, not an asymptotic complexity bound.

\textbf{Comparison with Alternative Inferential Methods:}
\begin{itemize}
    \item \textit{Delta Method:} Requires $\mathcal{O}(1)$ operations. However, its first-order asymptotic variance collapses when two path coefficients vanish (\ref{app:delta}), and the resulting near-boundary degradation produces severe undercoverage (Section~\ref{subsec:coverage}).
    \item \textit{Nonparametric Bootstrap:} Incurs an $\mathcal{O}(B \times N)$ computational cost, where $B \ge 5000$ resamples and sample size $N$ make model refitting computationally demanding in simulation studies and power calculations.
    \item \textit{Monte Carlo Integration:} Requires $\mathcal{O}(N_{\text{MC}})$ samples, where $N_{\text{MC}} \approx 10^8$ independent draws are necessary to match the accuracy of the proposed quadrature routine (Section~\ref{sec:simulation}).
    \item \textit{Proposed Dimension-Reduction Method:} Requires $\mathcal{O}(n^2)$ grid evaluations and roughly $0.9\text{--}5.0 \times 10^4$ adaptive evaluations at $\delta = 10^{-6}$ on the evaluated cases, evaluating the CDF in under one second (Table~\ref{tab:timing}) with mean absolute error roughly 5500-fold smaller than that of the Delta method (Table~\ref{tab:accuracy}).
\end{itemize}

\subsection{Convergence Analysis}
\label{subsec:convergence}

A key theoretical property of Equation~(\ref{eq:cdf_final}) is the integrability of the integrand near the coordinate axes where $xy=0$. While the quotient $v/(xy)$ diverges as $xy \to 0$ for $v \neq 0$, the normal CDF factor satisfies $0 \le \Phi(\cdot) \le 1$ everywhere. The integrand is therefore strictly bounded by the integrable bivariate Gaussian density $f_{X,Y}(x,y)$, guaranteeing that the CDF integral exists and is finite for all $v \in \mathbb{R}$. By the Dominated Convergence Theorem, numerical evaluation over expanding compact domains converges to the exact integral. Theorem~\ref{thm:convergence} formalizes this result, Corollary~\ref{cor:truncation} provides an explicit uniform truncation error bound, and \ref{app:convergence} details the complete proof.

\begin{theorem}[Finiteness and Truncation Convergence]
    \label{thm:convergence}
    Let $f_{X,Y}(x,y)$ be the marginal density of $(X,Y)$ under $\bs{X} \sim \mathcal{N}_3(\bs{\mu},\bs{\Sigma})$ with $\bs{\Sigma} \succ 0$, as defined in Section~\ref{sec:back}. For $xy \neq 0$, define
    \begin{equation}
        h_v(x,y) =
        \Phi\left(\text{sgn}(xy) \cdot \frac{v/(xy) - \muZ(x,y)}{\sigZ}\right)
        f_{X,Y}(x,y),
    \end{equation}
    and set $h_v(x,y)=0$ when $xy=0$. Then $h_v$ is integrable over $\mathbb{R}^2$, and
    \[
        I(v)=\iint_{\mathbb{R}^2} h_v(x,y)\,dy\,dx
    \]
    is finite for every $v \in \mathbb{R}$. Moreover, for
    \[
        \mathcal{D}_n
        = \{(x,y): |xy| \geq n^{-1},\ |x| \leq n,\ |y| \leq n\},
    \]
    the truncated integrals satisfy
    \[
        \lim_{n\to\infty}\iint_{\mathcal{D}_n}h_v(x,y)\,dy\,dx=I(v).
    \]
\end{theorem}

\section{Computational Implementation}
\label{sec:computation}

\subsection{Numerical Stability and Precision}
\label{subsec:stability}

\subsubsection{Standardization and Conditioning}

When $|\rho_{XY}| \to 1$, the $2 \times 2$ submatrix $\bs{\Sigma}_{XY}$ becomes ill-conditioned, which can destabilize the computation of the conditional regression vector $\bs{b}_{Z|XY}$. We implement two complementary numerical strategies to safeguard stability:

\textbf{(1) Variable Standardization:} As implemented in Algorithm~\ref{alg:cdf_computation} (Lines~\ref{ln:pre-first}--\ref{ln:pre-last}), we rescale the variables to unit variance by evaluating the correlation matrix $\bs{R} = \diag(\bs{\sigma})^{-1} \bs{\Sigma} \, \diag(\bs{\sigma})^{-1}$. This transformation eliminates marginal scale disparities, reducing the dynamic range of matrix elements, and can thereby improve numerical conditioning, although diagonal scaling does not guarantee a smaller condition number for every $\bs{\Sigma}$. The standardized threshold $v^* = v/(\sigma_X \sigma_Y \sigma_Z)$ and standardized mean vector $\bs{\mu}^*$ ensure exact scale invariance.

\textbf{(2) Cholesky Decomposition:} Rather than inverting $\bs{R}_{XY}$ directly, we solve the linear system $\bs{R}_{XY}\bs{b}_{Z|XY} = \bs{R}_{Z,XY}^{\T}$ by Cholesky factorization followed by forward and backward substitution (Algorithm~\ref{alg:cdf_computation}, Lines~\ref{ln:chol-first}--\ref{ln:chol-last}). This approach is backward stable for positive-definite systems and limits roundoff accumulation \citep{higham2002}. Because $\bs{R}_{XY} = \bs{L}\bs{L}^{\T}$, the eigenvalues of $\bs{R}_{XY}$ equal the squares of the singular values of $\bs{L}$, yielding the 2-norm condition number $\kappa_2(\bs{L}) = \sqrt{\kappa_2(\bs{R}_{XY})}$. Although the conditioning of the linear solve remains governed by $\kappa_2(\bs{R}_{XY})$, Cholesky factorization avoids explicit matrix inversion and is computationally cheaper.

\subsubsection{Near-Axis Numerical Safeguard}

The analytical formula in Equation~(\ref{eq:cdf_final}) contains the quotient $v/(xy)$, which is undefined along the coordinate axes $x=0$ and $y=0$. Although these axes have Lebesgue measure zero and do not contribute probability mass (\ref{app:convergence}), floating-point evaluation must return a finite value at points on or extremely near the axes, where $v/(xy)$ is infinite or, at $v=0$, the undefined $0/0$.

We implement a near-axis numerical safeguard (Algorithm~\ref{alg:cdf_computation}, Lines~\ref{ln:axis-first}--\ref{ln:axis-last}). When $|xy| < \epsilon_{\mathrm{axis}}$ with standardized cutoff $\epsilon_{\mathrm{axis}} = 10^{-10}$, we set the argument of $\Phi(\cdot)$ to $\text{sgn}(v^*) \cdot 10^{10}$. This threshold is independent of the quadrature error tolerances. For $xy \neq 0$, the exact argument is
\begin{equation}
    a_v(x,y)
    = \frac{\text{sgn}(xy)}{\sigZ}
      \left(\frac{v}{xy} - \muZ(x,y)\right) .
\end{equation}
For any point $(x_0, y_0)$ on the coordinate axes with $v \neq 0$,
\begin{equation}
    \lim_{\substack{(x,y)\to(x_0,y_0)\\xy\neq0}}
    \Phi\left(a_v(x,y)\right) =
    \begin{cases}
        1 & \text{if } v > 0 ,\\
        0 & \text{if } v < 0 ,
    \end{cases}
\end{equation}
regardless of the path of approach. When $v=0$, the safeguard assigns $\Phi(0)=1/2$ on the thin strip $|xy| < \epsilon_{\mathrm{axis}}$ as a numerical convention. Because the Gaussian CDF satisfies $0 \le \Phi(\cdot) \le 1$, the absolute integration error introduced by this substitution is bounded by $P(|XY| < \epsilon_{\mathrm{axis}})$. This probability vanishes as $\epsilon_{\mathrm{axis}} \to 0$ because $(X,Y)$ is continuously distributed. For the six scenarios of Section~\ref{sec:simulation}, computed by one-dimensional numerical integration in the standardized coordinates, it is at most $2.1 \times 10^{-9}$ (High Correlation), well below the $10^{-6}$ quadrature tolerance. Mathematically $\Phi(-10) = 7.6 \times 10^{-24}$, and in double precision $\Phi(10)$ already evaluates to exactly $1$; an argument of $\pm 10^{10}$ therefore returns the limiting values $1$ and $0$ to machine precision.

\subsubsection{Evaluation of the Bivariate Normal PDF}

Evaluating $f_{X,Y}(x,y)$ requires computing the exponential of a quadratic form. Far from the mean under high correlation $|\rho_{XY}| \approx 1$ (for example, at domain corners where $\rho_{XY} = 0.99$ and the exponent falls below $-6000$), the bivariate density underflows to zero in double precision. This underflow does not degrade accuracy because the corresponding tail regions carry negligible probability mass. The adaptive cubature routine naturally concentrates evaluations in regions with non-negligible probability density.

Numerical overflow cannot occur: for $\bs{z} = [x - \mu_X^*, y - \mu_Y^*]^{\T}$, the exponent quadratic form
\begin{equation}
    -\frac{1}{2}\bs{z}^{\T} \bs{R}_{XY}^{-1} \bs{z}
\end{equation}
is negative semi-definite for all $|\rho_{XY}| < 1$. Because the exponent is non-positive, $\exp(\cdot)$ can only underflow toward zero, eliminating the need for logarithmic-scale transformations.

\subsection{Integration Domain and Truncation}
\label{subsec:integration}

\subsubsection{Infinite Domain Truncation}

Although Equation~(\ref{eq:cdf_final}) defines an integral over $\mathbb{R}^2$, numerical cubature requires a bounded domain. We truncate integration to the symmetric box $\mathcal{D}_c = [\mu_X^* - c, \mu_X^* + c] \times [\mu_Y^* - c, \mu_Y^* + c]$ with standardized half-width $c = 8$ (Algorithm~\ref{alg:cdf_computation}, Line~\ref{ln:int-first}).

This choice is grounded in Corollary~\ref{cor:truncation}, which bounds the truncation error by the excluded marginal probability mass:
\begin{equation}
    \left| F_V(v) - \iint_{\mathcal{D}_c} g(x,y) \, dy \, dx \right| \;\leq\; 4\,\Phi(-c) ,
\end{equation}
uniformly for all $v \in \mathbb{R}$ and all positive-definite covariance matrices $\bs{\Sigma} \succ 0$. Setting $c = 8$ gives $4\Phi(-8) \approx 2.5 \times 10^{-15}$, which is more than eight orders of magnitude below the default relative tolerance of $10^{-6}$. Domain truncation error is therefore negligible in double precision, ensuring that any residual numerical error reflects quadrature discretization rather than boundary omission.

\subsubsection{Adaptive Cubature Strategy}
\label{subsubsec:cubature}

We implement the $h$-adaptive cubature method (\texttt{hcubature}) from the \texttt{cubature} R package \citep{cubature}, which executes globally adaptive domain subdivision \citep{berntsen1991}. The algorithm partitions the integration domain into hyper-rectangles, applying an embedded degree-7/degree-5 Genz--Malik cubature pair \citep{genzmalik1980}.

Two algorithmic features govern performance:
\begin{itemize}
    \item \textbf{Adaptive Error Bisection:} Subregions with the largest estimated local error are bisected along the coordinate axis with the largest fourth difference, concentrating evaluations where the integrand exhibits steep gradients.
    \item \textbf{Quadrant-Based Domain Partitioning:} Partitioning the domain at $x=0$ and $y=0$ aligns each subregion strictly within an open quadrant, ensuring that $\text{sgn}(xy)$ is constant within each cubature call.
\end{itemize}

A caveat applies to adaptive error control. That an automatic routine can return a value in error while reporting convergence is a documented property of this class of algorithms rather than an implementation-specific defect. \citet{berntsen1991} analyze this issue for routines that discard earlier function evaluations across subdivision levels (including the Genz--Malik algorithm \citep{genzmalik1980}), where error estimates based on finite evaluations can prove unreliable. They introduce a two-level error estimate in response. This phenomenon has been studied since \citet{lyness1976}; \citet{gonnet2012} characterizes the families of integrands that cause failure for various adaptive estimators. In a sweep of 1000 random correlation matrices evaluated at $v = 0$ with $\bs{\mu} = \mathbf{0}$, where $F_V(0) = \tfrac{1}{2}$ exactly by symmetry, \texttt{hcubature} returned values in error by up to $9.2 \times 10^{-4}$ on two matrices while reporting convergence. The error persisted under tolerance refinement to $10^{-10}$, and the internal error estimate understated the true error by up to three orders of magnitude. Both cases occurred among the seven matrices with condition number $\kappa_2(\bs{\Sigma}) > 10^{3}$ (2 of 1000 overall; 2 of 7 among the ill-conditioned). Ill-conditioning alone does not imply failure: the remaining five matrices showed only ordinary quadrature error (at most $3.6 \times 10^{-6}$) that vanished under refinement. For near-singular $\bs{\Sigma}$, we recommend cross-checking the default method against the fixed-node \texttt{"gauss"} option at two successive node counts, or against exact anchor values such as $F_V(0) = \tfrac{1}{2}$ at $\bs{\mu} = \mathbf{0}$. The six simulation scenarios of Section~\ref{sec:simulation} are unaffected: the cross-engine disagreement is below the $10^{-6}$ tolerance at every scenario and threshold checked (maximum $3.9 \times 10^{-7}$).

\subsection{Edge Cases and Pathological Scenarios}
\label{subsec:edge}

\subsubsection{Mixed-Sign Correlations}

When the correlations among the coefficient estimators have mixed signs (for instance, $\rho_{XY} > 0$, $\rho_{XZ} < 0$, $\rho_{YZ} > 0$), the product distribution departs from normality primarily through excess kurtosis rather than skewness. Mixed signs here are a property of the sampling covariance, not of the path coefficients: in the Mixed-Sign Correlation scenario every coefficient is zero ($\bs{\mu} = \bs{0}$), so the scenario does not represent inconsistent mediation, which requires nonzero coefficients. Under the central null $\bs{\mu} = \bs{0}$, the joint Gaussian density is invariant under $(X,Y,Z) \mapsto (-X,-Y,-Z)$, which reflects $V \mapsto -V$. The distribution of $V$ is therefore symmetric about zero, and all odd central moments vanish. Non-normality manifests as heavy tails: the excess kurtosis of $V$, derived from Gaussian product moments (Remark~\ref{rem:moments}, \ref{app:convergence}), is $32.5$ in the Mixed-Sign Correlation scenario and ranges from $24.0$ ($\bs{\Sigma} = \mathbf{I}_3$) to $50.6$ across zero-mean configurations. The Delta method's error in this regime is comparable to other scenarios (Table~\ref{tab:accuracy}); its failure is driven by gradient collapse near zero coefficients (\ref{app:delta}) rather than tail weight or correlation patterns.

The proposed dimension-reduction method requires no special modifications under mixed-sign correlations. The conditional probability retains its standard Gaussian form within each quadrant, and the conditional mean $\muZ(x,y)$ incorporates the mixed-sign covariance structure directly through Equation~(\ref{eq:cond_mean}).

\subsubsection{Near-Zero Means}

When two or more path coefficients approach zero, the Delta-method gradient collapses toward zero (\ref{app:delta}), producing maximum CDF errors near $\maxDeltaZeroMean$ (Table~\ref{tab:accuracy}). The proposed method requires no special adjustments: Theorem~\ref{thm:main} and the conditional Gaussian distribution hold for arbitrary mean vectors $\bs{\mu}$, and Algorithm~\ref{alg:cdf_computation} executes uniformly across all parameter regimes.

\subsubsection{Perfect or Near-Perfect Correlation}

When $|\rho_{XY}| > 0.99$, the $2 \times 2$ correlation submatrix $\bs{R}_{XY}$ becomes ill-conditioned, with $\kappa_2(\bs{R}_{XY}) = (1+|\rho_{XY}|)/(1-|\rho_{XY}|) > 199$. Ill-conditioning of this block alone does not make the conditional variance small: with $\rho_{XZ} = \rho_{YZ} = 0$, $\bs{\Sigma}$ approaches singularity as $|\rho_{XY}| \to 1$ while $\sigZ^2 = \sigma_Z^2$ stays fixed. Because $\det \bs{\Sigma} = \det \bs{\Sigma}_{XY} \cdot \sigZ^2$, the full covariance matrix approaches singularity by two distinct routes. If $\bs{\Sigma}_{XY}$ approaches singularity, the solve for $\muZ$ (Equation~(\ref{eq:cond_mean})) becomes ill-conditioned. If instead $Z$ approaches a linear function of $(X, Y)$, the Schur complement $\sigZ^2$ (Equation~(\ref{eq:cond_var})) vanishes and the conditional distribution of $Z$ degenerates.

We safeguard against singularity by verifying that all eigenvalues of $\bs{\Sigma}$ are strictly positive ($\lambda_{\min} > 0$). If any eigenvalue is non-positive ($\lambda_{\min} \le 0$), the function terminates with an informative error message. For positive-definite but ill-conditioned matrices ($\kappa_2(\bs{\Sigma}) > 10^3$), the adaptive cubature diagnostic cross-checks of Section~\ref{subsubsec:cubature} apply.

\subsection{Software Implementation and Interface}
\label{subsec:software}

Our R implementation provides an accessible interface for sequential mediation analysis. The primary function \texttt{p\_prod3\_v1(q, mean, cov, method = "hcubature")} takes a scalar threshold \texttt{q}, a length-3 mean vector \texttt{mean}, a $3 \times 3$ positive-definite covariance matrix \texttt{cov}, and an integration \texttt{method} (\texttt{"hcubature"} or \texttt{"gauss"}). It returns the scalar cumulative probability $F_V(q)$. For numerical audits, setting \texttt{return\_diagnostics = TRUE} returns per-cell error estimates, function evaluation counts, and convergence flags without altering the core return structure. Confidence intervals are obtained by inverting the distribution function, identifying $\{q : F_V(q) \in [\alpha/2, 1 - \alpha/2]\}$ via univariate root-finding.

\textbf{Computational Performance:} On an Apple M2 Pro laptop (16GB RAM), evaluating a single CDF value requires approximately \timeEvalLaptopLo--\timeEvalLaptopHi{} seconds across diverse correlation regimes. Inverting the CDF via root-finding to construct a 95\% confidence interval requires \nEvalCILaptopLo--\nEvalCILaptopHi{} function evaluations (two \texttt{uniroot} inversions at tolerance $10^{-6}$), taking \timeCILaptopLo--\timeCILaptopHi{} seconds. By comparison, nonparametric bootstrap confidence intervals with $B = 5000$ resamples require \timeBootstrapEst{} seconds on the same hardware, representing an order-of-magnitude reduction in computation time (\timeCILaptopLo--\timeCILaptopHi{} versus \timeBootstrapEst{} seconds per interval).

\textbf{Independent Cross-Validation:} We validated \texttt{p\_prod3\_v1()} against the orthant probability algorithm of \citet{genz2009} in the \texttt{mvtnorm} package. At $q=0$, the probability $P(XYZ < 0)$ decomposes into the sum of four disjoint orthant probabilities (three containing exactly one negative factor and one containing three negative factors). Evaluated across all six simulation scenarios, this orthant decomposition agrees with \texttt{p\_prod3\_v1()} to within $5 \times 10^{-7}$. An independent Monte Carlo cross-check using $N = 10^7$ multivariate normal draws at non-zero thresholds ($q \in \{-1, -0.5, 0.5, 1\}$) confirms agreement to within $3 \times 10^{-4}$, within two Monte Carlo standard errors. The validation code is provided in \texttt{code/R/independent\_crosscheck.R}; its per-scenario results are archived in \texttt{simulation\_results/independent\_crosscheck.csv}.

\section{Simulation Studies}
\label{sec:simulation}

We conducted Monte Carlo simulations across six representative scenarios (each evaluated at 99 quantiles with 1000 replications; $n = 594{,}000$ scenario--quantile--replication observations) to evaluate the numerical accuracy and computational efficiency of the proposed dimension-reduction method. We compared our approach against two standard alternatives in mediation analysis: the Delta method and parametric Monte Carlo simulation.

\subsection{Simulation Design}

We examined six scenarios spanning independence to strong correlation ($\rho \in \{0, 0.3, 0.7\}$), mixed-sign correlations, non-zero means, and unequal variances:

\begin{itemize}
    \item \textbf{Scenario 1 (Independent):} Independent standard normal variables ($\bs{\mu} = \mathbf{0}, \bs{\Sigma} = \mathbf{I}_3$).
    \item \textbf{Scenario 2 (Moderate Correlation):} Symmetric correlation structure with $\rho = 0.3$.
    \item \textbf{Scenario 3 (High Correlation):} Strong positive correlations ($\rho = 0.7$).
    \item \textbf{Scenario 4 (Mixed-Sign Correlation):} Mixed signs in correlations ($\rho_{XY} = 0.5, \rho_{XZ} = -0.3, \rho_{YZ} = 0.2$) with $\bs{\mu} = \mathbf{0}$. All three coefficients are zero, so this scenario varies the correlation pattern only and does not represent inconsistent mediation.
    \item \textbf{Scenario 5 (Non-Zero Means):} $\bs{\mu} = [0.5, 0.3, 0.2]^{\T}$, with uniform correlation $\rho = 0.2$.
    \item \textbf{Scenario 6 (Unequal Variances):} Heterogeneous variances ($\sigma_X^2 = 2.0, \sigma_Y^2 = 1.0, \sigma_Z^2 = 0.5$) with covariances $\sigma_{XY} = 0.4$, $\sigma_{XZ} = 0.3$, $\sigma_{YZ} = 0.2$.
\end{itemize}

\subsection{Evaluation Metrics and Benchmark}

To establish a high-precision reference benchmark, we computed cumulative probabilities using $N = 10^8$ Monte Carlo draws per scenario and quantile. For an individual CDF estimate, the worst-case standard error of the indicator average is $1/(2\sqrt{N}) = 5 \times 10^{-5}$, with exact standard error depending on the underlying probability. The reference contains Monte Carlo sampling variability and serves as an empirical benchmark rather than a deterministic upper bound.

We quantified accuracy via the mean absolute error (MAE) across 99 quantiles spanning $q \in \{0.01, 0.02, \dots, 0.99\}$. Each scenario--quantile cell was evaluated over 1000 replications: the parametric Monte Carlo comparator was redrawn with a deterministic seed and timings were recorded on every replication, whereas the proposed method and the Delta method are deterministic given $(\bs{\mu}, \bs{\Sigma})$.

For the parametric Monte Carlo comparator, we drew 5000 samples directly from $(\bs{\mu}, \bs{\Sigma})$ and computed the empirical CDF, matching the resolution commonly used for Monte Carlo confidence intervals in mediation analysis \citep{preacher2008}. The Delta method served as an instantaneous baseline. The simulation evaluates the diagonal-only form of $\bs{\Sigma}_\beta$ as its baseline; \ref{app:delta} provides the full covariance formulation.

\subsection{Results}

\subsubsection{Numerical Accuracy}

Table~\ref{tab:accuracy} reports the per-scenario maximum absolute error for each method alongside overall summary statistics. Across all scenarios and quantiles ($n = 594{,}000$ observations), the proposed dimension-reduction method achieved a mean absolute error of $\maeProposed$ relative to the $10^8$-sample benchmark. Because replications within each cell share a cached reference value, benchmark uncertainty is governed by the reference sampling error ($\approx 3.0 \times 10^{-5}$ at $N = 10^{8}$) rather than a row-level standard error of the mean. The maximum observed error across all evaluations was $1.3 \times 10^{-4}$.

In contrast, the Delta method exhibited substantial error: its mean absolute error of $0.1717$ is roughly 5500-fold larger than that of the proposed method, consistent with its documented limitations in mediation analysis \citep{mackinnon2002}. The maximum error reached $0.4548$ in the Independent scenario ($\bs{\mu} = \mathbf{0}$) and was similarly elevated ($0.44\text{--}0.45$) across all other zero-mean scenarios, including Mixed-Sign Correlation. In contrast, the non-zero mean scenario (Non-Zero Means, $\bs{\mu} = [0.5, 0.3, 0.2]^{\T}$) reduced that maximum by more than half, to $0.2030$. This pattern confirms that large approximation errors occur whenever multiple path coefficients approach zero, as established by the gradient-collapse theory in \ref{app:delta}.

The parametric Monte Carlo method ($5000$ draws) achieved a mean absolute error of $0.0042$, representing a 41-fold improvement over the Delta method but remaining 136-fold less accurate than the proposed approach. Its accuracy is constrained by Monte Carlo sampling error, which decreases at the standard $1/\sqrt{B}$ rate.

\begin{table}[ht]
\centering
\caption{Maximum Absolute Error in CDF Estimation by Scenario\CAPNOTE{Per-scenario \emph{maximum} absolute error; the headline $\maeProposed$ is the overall \emph{mean} absolute error (MAE) across scenarios and quantiles---the two summaries differ by design.}}
\label{tab:accuracy}
\small
\setlength{\tabcolsep}{4pt}
\begin{tabular}{lccc}
\toprule
Scenario & Proposed Method & Delta Method & Parametric MC ($5000$ draws) \\
\midrule
Independent & $1.3 \times 10^{-4}$ & $0.455$ & $0.029$ \\
Moderate Correlation & $9.8 \times 10^{-5}$ & $0.452$ & $0.026$ \\
High Correlation & $1.1 \times 10^{-4}$ & $0.437$ & $0.029$ \\
Mixed-Sign Correlation & $1.0 \times 10^{-4}$ & $0.450$ & $0.027$ \\
Non-zero Means & $1.0 \times 10^{-4}$ & $0.203$ & $0.029$ \\
Unequal Variances & $1.1 \times 10^{-4}$ & $0.452$ & $0.027$ \\
\bottomrule
\end{tabular}
\end{table}

Figure~\ref{fig:results} examines the Mixed-Sign Correlation scenario in detail. Panel (a) shows the distribution of absolute errors on a base-10 logarithmic scale across more than four decades. The three methods separate into distinct bands: the proposed method lowest (errors near $10^{-5}$--$10^{-4}$), parametric Monte Carlo in the middle (near $10^{-3}$--$10^{-2}$), and the Delta method highest (near $10^{-1}$), reflecting the numerical differences in Table~\ref{tab:accuracy}. Panel (b) traces mean absolute error across quantile levels: the proposed method maintains stable error near $10^{-5}$, whereas the Delta method error remains elevated throughout the distribution ($0.20$ in the lower tail to $0.09$ in the upper tail), and parametric Monte Carlo remains near $3\text{--}5 \times 10^{-3}$. Panel (c) plots signed CDF error (estimate minus benchmark): the proposed and parametric Monte Carlo curves remain centered on zero, while the Delta method error reaches $0.45$ in the lower tail and $0.15$ at the median. Panel (d) displays the proposed method's mean wall-clock runtime per CDF evaluation across scenarios ($0.44\text{--}0.53$ seconds on the cluster hardware of Table~\ref{tab:timing}).

\begin{figure}[htbp]
\centering
\includegraphics[width=\textwidth]{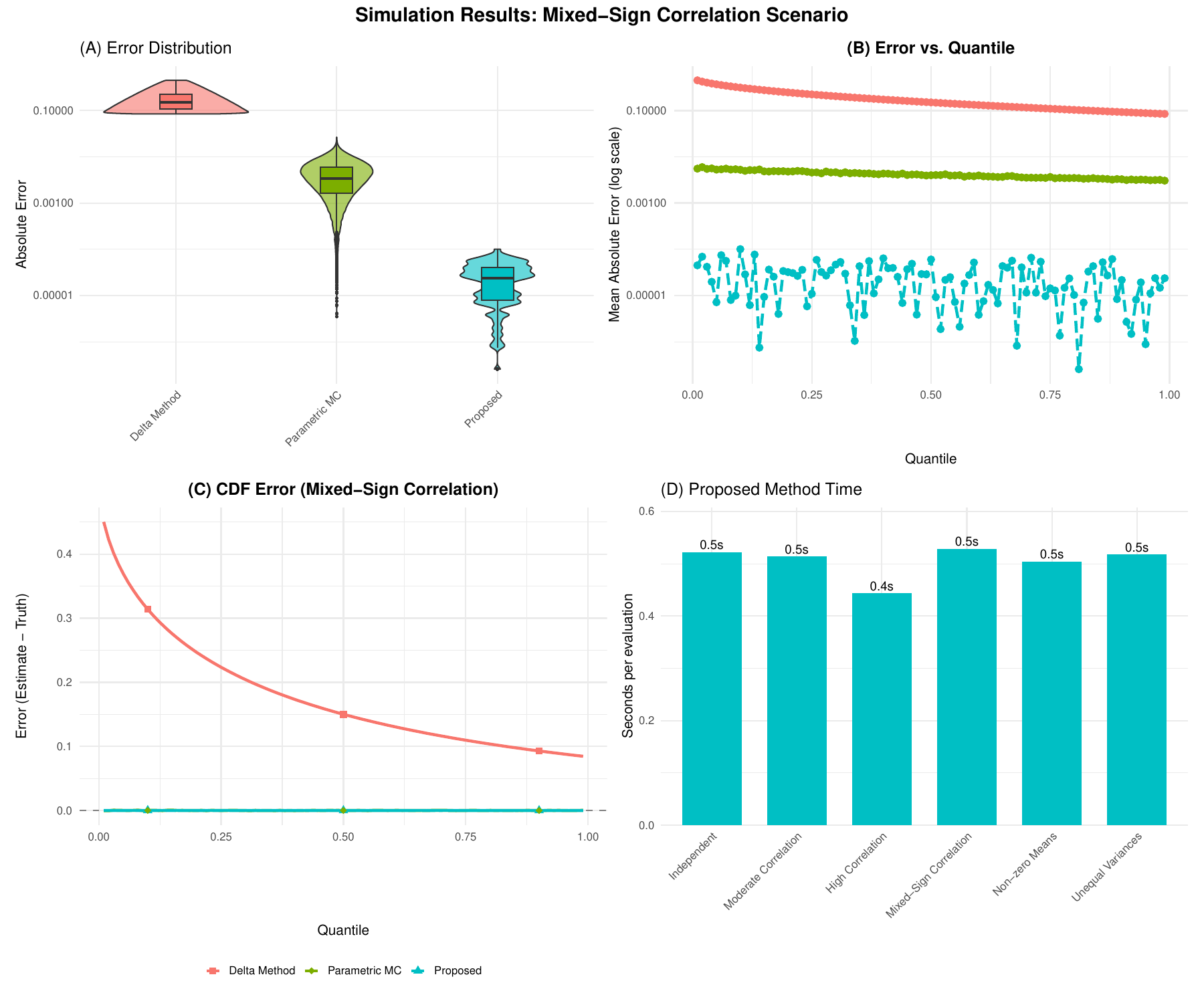}
\caption{Comparison of the proposed dimension-reduction method, the Delta method, and parametric Monte Carlo for the Mixed-Sign Correlation scenario (panels a--c) and timing across all six scenarios (panel d). (a) Distribution of absolute errors on a log scale: the proposed method attains mean errors below $10^{-4}$ while the Delta method exhibits large approximation error (overall mean absolute error $\maeDelta$). (b) Mean absolute error by quantile level. (c) Signed CDF error (estimate $-$ benchmark) by quantile level. (d) Proposed method's mean time per CDF evaluation by scenario; cross-method timings appear in Table~\ref{tab:timing}.}
\label{fig:results}
\end{figure}

\subsubsection{Computational Efficiency}

Table~\ref{tab:timing} presents the computational cost for each method, measured on the UNM CARC Hopper cluster (see Section~\ref{subsec:software} for standalone laptop benchmarks). The proposed dimension-reduction method required an average of 0.51 seconds per CDF evaluation on this cluster environment. While slower than the Delta method ($< 1$ millisecond), the proposed method is computationally efficient compared to procedures of matching accuracy: obtaining equivalent precision via Monte Carlo integration requires $\approx 10^8$ draws, and nonparametric bootstrap confidence intervals require \timeBootstrapEst{} seconds on laptop hardware (Section~\ref{subsec:software}).

The parametric Monte Carlo comparator in our simulation drew directly from the known model and required only milliseconds per evaluation. In applied settings, however, simulation- and resampling-based confidence intervals scale linearly with the number of draws or resamples $B$, a cost that compounds in iterative procedures such as power analysis or Bayesian posterior predictive checks: at the measured \timeBootstrapEst{} seconds per bootstrap interval (Section~\ref{subsec:software}), a power analysis with $10^4$ replications would require roughly 10--28 days of serial computation. For comparison, evaluating the CDF at 100 quantiles requires approximately 51 seconds with the proposed method, about one minute of computation for a full distributional summary.

\begin{table}[ht]
\centering
\caption{Computational Time Comparison by Method and Scenario\CAPNOTE{Per-CDF mean wall time; Delta is an analytic variance formula (sub-millisecond), so it is not a speed competitor --- the proposed method's value is accuracy, not speed.}}
\label{tab:timing}
\begin{tabular}{lccc}
\toprule
Scenario & Proposed & Delta & Parametric MC \\
\midrule
Independent & $522.1$ ms & $<0.1$ ms & $1.5$ ms \\
Moderate Correlation & $513.9$ ms & $<0.1$ ms & $1.5$ ms \\
High Correlation & $443.8$ ms & $<0.1$ ms & $1.5$ ms \\
Mixed-Sign Correlation & $528.5$ ms & $<0.1$ ms & $1.5$ ms \\
Non-zero Means & $504.2$ ms & $<0.1$ ms & $1.5$ ms \\
Unequal Variances & $518.0$ ms & $<0.1$ ms & $1.5$ ms \\
\bottomrule
\end{tabular}
\end{table}

\subsubsection{Practical Implications}

For applied research, these findings support four conclusions:
\begin{enumerate}
    \item The Delta method should be avoided for products of three coefficients when two or more coefficients may be small or near zero, where gradient collapse causes substantial undercoverage and CDF errors exceeding $0.17$.

    \item The proposed method agrees with the $10^8$-sample benchmark to a mean absolute discrepancy of $\maeProposed$ (maximum observed $1.3 \times 10^{-4}$), comparable to the benchmark's own sampling error ($\approx 3.0 \times 10^{-5}$); these are observed discrepancies from a Monte Carlo reference, not a guaranteed error bound. This accuracy supports hypothesis tests and confidence intervals at conventional significance levels ($\alpha = 0.05$); the coverage study of Section~\ref{subsec:coverage} evaluates the induced intervals directly.

    \item Evaluating the CDF across 19 quantiles for distributional plotting and interval formation completes in under one minute on standard hardware.

    \item Parametric Monte Carlo serves as a useful diagnostic tool but is inefficient as a primary estimation method in large-scale simulation or iterative workflows.
\end{enumerate}

The proposed method's per-scenario maximum absolute error remained between $9.8 \times 10^{-5}$ and $1.3 \times 10^{-4}$ across all six scenarios (Table~\ref{tab:accuracy}), including the heavy-tailed Mixed-Sign Correlation scenario (excess kurtosis $32.5$; Section~\ref{subsec:edge}). This consistency supports inference in sequential mediation models across the regimes studied here; the one known exception is near-singular $\bs{\Sigma}$, where the adaptive error control can fail silently and the cross-checks of Section~\ref{subsubsec:cubature} remain necessary.

\subsection{Coverage of the Induced Confidence Intervals}
\label{subsec:coverage}

We evaluate the empirical coverage of 95\% plug-in quantile intervals: the $\alpha/2$ and $1-\alpha/2$ quantiles of the product distribution evaluated at the estimated mean vector and covariance matrix, computed with the proposed CDF. The simulation adheres to the ADEMP framework \citep{morris2019}. Coefficient estimates were drawn from the asymptotic sampling distribution $\hat{\bs{\beta}} \sim \mathcal{N}_3(\bs{\beta}, \bs{\Sigma}_{\beta}/N)$, where $\hat{\bs{\beta}} = (\hat\beta_1, \hat\beta_2, \hat\beta_3)^{\T}$ is the path-coefficient estimator vector (\ref{app:delta}). The factorial design crosses the six simulation scenarios of Table~\ref{tab:accuracy} and a regular anchor scenario ($\bs{\mu} = (1,1,1)^{\T}$, $\theta = 1$) with three sample sizes ($N \in \{50, 200, 1000\}$) and two covariance structures (known $\bs{\Sigma}_{\beta}$ versus sample-estimated Wishart covariance), yielding \covCells{} cells. Each cell was evaluated over $B = \covB$ replicates, yielding a Monte Carlo standard error of at most \covMCSE{} at nominal $95\%$ coverage. We compared four interval methods across all replicates: the proposed plug-in quantile interval obtained via root-finding (Section~\ref{subsec:software}), the Monte Carlo interval of \citet{mackinnon2004} with $10^4$ draws, a projected-Wald ellipsoid interval (a conservative Scheff\'e-type construction included for validity corroboration), and the first-order Delta-method Wald interval. Replications used deterministic seeds; total execution required approximately \covCoreHours{} core-hours, with a numerical failure rate under \covFailMaxPct{}\%. Table~\ref{tab:coverage} reports coverage under the conservative failure-as-miss convention.

Three principal findings emerge from Table~\ref{tab:coverage}:
\begin{enumerate}
    \item \textbf{Boundary Robustness:} At the non-regular boundary point $\bs{\mu} = \bs{0}$, where the first-order asymptotic variance collapses (\ref{app:delta}), the plug-in quantile interval achieved empirical coverage between \covTImuZeroLo{} and \covTImuZeroHi{} across all thirty zero-mean cells, above the nominal level in every cell. The ellipsoid benchmark corroborates this validity ($[\covElMuZeroLo, \covElMuZeroHi]$ across the same cells). The Monte Carlo interval of \citet{mackinnon2004}, which takes quantiles of the same plug-in distribution by simulation rather than through the proposed CDF, shows the same elevated coverage (the two differ by at most \covMCTIcovDiff{} in every cell; see below), so the elevation is unlikely to be a numerical artifact of the CDF. We conjecture that it reflects the non-regularity of the product at $\bs{\mu} = \bs{0}$ but do not prove this mechanism here.
    \item \textbf{Asymptotic Regularity:} Outside the zero-mean boundary, the observed coverage of the plug-in quantile interval was near nominal at every sample size, including $N = 50$ (between \covTIminAll{} and \covTInbHi{}). No cell fell materially below nominal: the minimum, \covTIminAll{} (regular anchor, $N = 50$, estimated covariance), has Clopper--Pearson interval $[\covTIminAllCPlo, \covTIminAllCPhi]$, which contains $0.95$. In contrast, the Delta method degrades at both extremes of the parameter space. Under the sharp null ($\bs{\mu} = \bs{0}$), it covered in all \covB{} replicates due to test-statistic bounding (Remark~\ref{rem:plugin}). Near the boundary (the Non-Zero Means scenario, $\theta = 0.03$), gradient collapse manifests as severe undercoverage (\covDeltaNBfifty{} at $N = 50$), recovering to \covDeltaNBthousand{} only at $N = 1000$. The plug-in quantile interval's observed coverage stays near the nominal level (\covTInbLo--\covTInbHi{}) across all sample sizes.
    \item \textbf{Monte Carlo Concordance:} Plug-in quantile intervals closely match the Monte Carlo intervals of \citet{mackinnon2004}: empirical coverage differed by at most \covMCTIcovDiff, mean interval widths differed by at most \covMCTIwidthPct{}\%, and proper interval scores \citep{gneiting2007} differed by at most \covMCTIscorePct{}\%. Both methods take quantiles of the same plug-in product distribution; the proposed interval evaluates them by numerical integration rather than simulation, which removes Monte Carlo noise from the endpoints.
\end{enumerate}

\begin{table}[t]
\centering
\caption{Empirical coverage of nominal 95\% confidence intervals in the
coverage study (\covCells{} cells, $B = \covB$ replicates per cell;
coverage Monte-Carlo standard error $\le \covMCSE$ at $p = 0.95$).
Values follow the worst-case convention, counting numerical failures as
misses (maximum per-cell failure rate $\covFailMaxPct\%$); complete-case
coverage was also tabulated and agrees within $\covCCWCmax$ in every cell.
Zero mean rows report the minimum--maximum across the five zero mean
scenarios. The projected-Wald ellipsoid benchmark, conservative by
construction and used for validity corroboration only, covered within
$[\covElLo, \covElHi]$ across all cells. Clopper--Pearson intervals for
every cell are given in the coverage summary deposited with the data (see the Data availability statement).}
\label{tab:coverage}
\footnotesize
\setlength{\tabcolsep}{4pt}
\begin{tabular}{@{}llcccccc@{}}
\toprule
 & & \multicolumn{2}{c}{Plug-in quantile} & \multicolumn{2}{c}{Monte Carlo} & \multicolumn{2}{c}{Delta} \\
\cmidrule(lr){3-4}\cmidrule(lr){5-6}\cmidrule(lr){7-8}
Scenario class & $N$ & known & est. & known & est. & known & est. \\
\midrule
Zero mean & 50 & 0.996--1.000 & 0.990--1.000 & 0.995--1.000 & 0.990--1.000 & 1.000 & 1.000 \\
 & 200 & 0.993--1.000 & 0.994--1.000 & 0.992--1.000 & 0.994--1.000 & 1.000 & 1.000 \\
 & 1000 & 0.992--1.000 & 0.993--1.000 & 0.991--1.000 & 0.993--1.000 & 1.000 & 1.000 \\
\addlinespace
Near boundary & 50 & 0.945 & 0.943 & 0.945 & 0.941 & 0.857 & 0.860 \\
\quad ($\theta = 0.03$) & 200 & 0.955 & 0.948 & 0.954 & 0.947 & 0.927 & 0.918 \\
 & 1000 & 0.949 & 0.950 & 0.950 & 0.948 & 0.947 & 0.947 \\
\addlinespace
Strong-regular anchor & 50 & 0.954 & 0.942 & 0.952 & 0.941 & 0.943 & 0.931 \\
\quad ($\theta = 1$) & 200 & 0.945 & 0.952 & 0.944 & 0.951 & 0.943 & 0.949 \\
 & 1000 & 0.943 & 0.946 & 0.940 & 0.944 & 0.942 & 0.945 \\
\bottomrule
\end{tabular}
\end{table}

\section{Discussion}
\label{sec:discussion}

The dimension-reduction method developed here extends bivariate normal product theory \citep{cui2016} to the trivariate cumulative distribution function. While \citet{ogasawara2025} characterized the density of $K$ correlated zero-mean normal variables, computing the cumulative distribution function has remained an open problem. By leveraging the conditional Gaussian distribution of $Z$ given $(X,Y)$, we reduce three-dimensional integration over a non-convex product-threshold region to adaptive two-dimensional quadrature over four smooth quadrants.

For applied researchers conducting sequential mediation analysis, this algorithm provides a computationally efficient alternative to the Delta method and nonparametric bootstrapping; under joint normality, it introduces no approximation to the product distribution. Our simulation results confirm that the Delta method is least reliable not when correlations are strong or mixed in sign, but when path coefficients approach zero. Every zero-mean scenario in Table~\ref{tab:accuracy} exhibited maximum error near $\maxDeltaZeroMean$, whereas the non-zero mean scenario \deltaNZMReduction, to $\maxDeltaNonZeroMean$. This pattern empirically validates the gradient-collapse theory in \ref{app:delta}: the first-order asymptotic variance collapses to zero whenever any two path coefficients vanish, irrespective of skewness or correlation structure. Nonparametric bootstrapping offers asymptotically valid inference but incurs heavy computational overhead (\timeBootstrapEst{} seconds per interval at $B = 5000$), compared to \timeCILaptopLo--\timeCILaptopHi{} seconds for the proposed plug-in quantile interval (Section~\ref{subsec:software}). The proposed approach delivers numerical accuracy ($\text{MAE} = \maeProposed$) at routine evaluation speed (\timeEvalLaptopLo--\timeEvalLaptopHi{} seconds per CDF value).

\subsection{Limitations and Extensions}

Four considerations define the scope of these results and suggest future directions:

First, the method is formulated for products of three variables ($K=3$). A chain with $m$ mediators has $K = m + 1$ path coefficients, and in longer chains ($m > 2$, so $K > 3$), conditional marginalization reduces the integral dimension to $K-1 = m$, coinciding with the dimensionality of the density in \citet[Theorem~1]{ogasawara2025}. Evaluating the resulting $(K-1)$-dimensional integral would require higher-dimensional cubature routines or quasi-Monte Carlo integration.

Second, we assume jointly normal coefficient estimators. While maximum-likelihood and least-squares estimators are asymptotically normal under standard regularity conditions, finite-sample non-normality or heavy-tailed mediators may require non-Gaussian extensions. The product-Jacobian representation extends to elliptical distributions \citep[Theorem~3]{ogasawara2025}, and developing conditional dimension reduction for elliptical families remains an open avenue.

Third, the algorithm evaluates the product distribution for a fixed parameter vector $(\bs{\mu}, \bs{\Sigma})$. In practice, these parameters are estimated with sampling error; incorporating second-order corrections or Bayesian posterior integration represents a valuable extension.

Fourth, identifying the sequential indirect effect as $\theta = \beta_1\beta_2\beta_3$ assumes linear additive structural equations. For models with categorical mediators or nonlinear interactions, causal indirect effects must be defined via counterfactual potential outcomes \citep{vanderweele2014c, valeri2013}, and in that setting the triple-product estimand (and hence this method) applies only when those effects reduce, under a specified model, to a product of jointly normal coefficient estimators; extending the method to categorical or nonlinear settings requires such a specification and is left to future work.

\subsection{Software Integration}

We recommend a straightforward two-step workflow for applied researchers analyzing two-mediator sequential chains ($A \to M_1 \to M_2 \to Y$). First, estimate the path coefficients and their asymptotic covariance matrix using structural equation modeling or linear regression. Second, pass the point estimates $\hat{\bs{\mu}}$ and covariance matrix $\hat{\bs{\Sigma}}$ to \texttt{p\_prod3\_v1()} to evaluate plug-in tail probabilities, such as $F_V(0)$ under the estimated distribution, or to form plug-in quantile confidence intervals. These tail probabilities are not exact $p$-values: they are computed at the estimated parameters rather than under a hypothesized null.

An interpretative consideration concerns variable scaling: because $\theta = \beta_1\beta_2\beta_3$ is a product of path coefficients, rescaling an intermediate mediator leaves the product unchanged as units cancel across adjacent links. When comparing indirect effect sizes across studies or metrics, researchers should standardize path coefficients before interpreting relative effect magnitudes \citep{preacher2011a}.

\section{Conclusion}
\label{sec:conclusion}

Sequential mediation analysis investigates mechanisms transmitted through two intermediate variables, targeting the three-path indirect effect $\theta = \beta_1\beta_2\beta_3$. Existing inference methods approximate the distribution of the estimator $\hat\theta = \hat\beta_1\hat\beta_2\hat\beta_3$ through first-order Taylor expansion, resampling, simulation, likelihood-ratio tests, or Bayesian posteriors \citep{taylor2008, tofighi2020}. We developed a dimension-reduction algorithm that evaluates the CDF through an analytic integral representation under joint normality, reducing the problem from a three-dimensional integral to an adaptive two-dimensional quadrature.

The proposed method achieves a mean absolute error of $\maeProposed$ relative to a $10^8$-sample Monte Carlo reference benchmark across diverse covariance structures. In the 42-cell empirical coverage study (Section~\ref{subsec:coverage}), plug-in quantile intervals attained observed coverage near nominal away from the zero-mean boundary and above nominal at it, with no cell materially below nominal at any sample size (minimum \covTIminAll, with a Clopper--Pearson interval that contains $0.95$), resolving the boundary undercoverage that undermines the Delta method in small-to-moderate samples (\covDeltaNBfifty{} at $N = 50$). Evaluating the distribution function in \timeEvalLaptopLo--\timeEvalLaptopHi{} seconds on standard hardware, this approach provides applied researchers with an accurate, computationally efficient foundation for statistical inference in sequential mediation models.

\section*{CRediT authorship contribution statement}
\textbf{Davood Tofighi:} Conceptualization, Methodology, Software, Formal analysis, Investigation, Visualization, Writing -- original draft, Writing -- review \& editing.

\section*{Declaration of competing interest}
The author declares that he has no known competing financial interests or personal relationships that could have appeared to influence the work reported in this paper.

\section*{Funding}
This research did not receive any specific grant from funding agencies in the public, commercial, or not-for-profit sectors.

\section*{Data availability}
The simulation results (594{,}000 per-replication records, the Monte Carlo ground-truth values, and the coverage-study summary), the R implementation, and all code needed to reproduce the results, figures, and tables are deposited in Mendeley Data under DOI~10.17632/vdc44jyj5r.1. No empirical data were used; all reported results derive from simulation.

\section*{Acknowledgements}
The author thanks the University of New Mexico Center for Advanced Research Computing (CARC), supported in part by the National Science Foundation, for the computational resources used in this work.

\section*{Declaration of generative AI and AI-assisted technologies in the writing process}
During the preparation of this manuscript the author used AI-assisted language tools for grammar and language editing. The author reviewed and edited all content and takes full responsibility for the content of this publication. No AI tool is listed as an author.

\appendix

\section{Proof of Theorem~\ref{thm:convergence}}
\label{app:convergence}

Denote the integrand by $h(x, y; v)$, defined for all $(x, y) \in \mathbb{R}^2 \setminus \mathcal{S}$, where $\mathcal{S} = \{(x, y) : xy = 0\}$ represents the coordinate axes.

On $\mathbb{R}^2 \setminus \mathcal{S}$, the factors of $h$ are continuous: $\Phi$ and $f_{X,Y}$ are continuous on $\mathbb{R}^2$, $\text{sgn}(xy)$ is locally constant off $\mathcal{S}$, and $\muZ(x,y)$ is affine in $(x, y)$. Because $\bs{\Sigma} \succ 0$, the conditional standard deviation satisfies $\sigZ > 0$ and is spatially constant. The product $h$ is therefore continuous on $\mathbb{R}^2 \setminus \mathcal{S}$ and Borel measurable. Because $0 \leq \Phi(t) \leq 1$ for all $t \in \mathbb{R}$, we obtain the domination bound:
\begin{equation}
    \label{eq:domination}
    |h(x, y; v)| \;\leq\; f_{X,Y}(x, y),
    \qquad \iint_{\mathbb{R}^2} f_{X,Y}(x,y) \, dy \, dx = 1 < \infty ,
\end{equation}
which establishes that $h \in L^1(\mathbb{R}^2 \setminus \mathcal{S})$. Consequently, $I(v)$ exists and is finite for every $v \in \mathbb{R}$. The excluded set $\mathcal{S}$ has two-dimensional Lebesgue measure zero and carries zero probability mass under the absolutely continuous density $f_{X,Y}$; undefined values of $v/(xy)$ along $\mathcal{S}$ do not alter the value of the integral.

The Dominated Convergence Theorem establishes convergence along any sequence of compact sets $\mathcal{D}_n$ that exhausts $\mathbb{R}^2 \setminus \mathcal{S}$ (for instance, $\mathcal{D}_n = \{(x,y) : |xy| \geq 1/n,\ |x| \leq n,\ |y| \leq n\}$). Because $h \mathbb{1}_{\mathcal{D}_n} \to h$ pointwise off $\mathcal{S}$ and $|h \mathbb{1}_{\mathcal{D}_n}| \leq f_{X,Y} \in L^1(\mathbb{R}^2)$, we have:
\begin{equation}
    \lim_{n \to \infty} \iint_{\mathcal{D}_n} h(x, y; v) \, dy \, dx = \iint_{\mathbb{R}^2} h(x, y; v) \, dy \, dx = I(v) . \qquad \qed
\end{equation}

The domination bound (\ref{eq:domination}) directly bounds the truncation error induced by restricting numerical integration to a compact domain.

\begin{corollary}[Truncation Error]
    \label{cor:truncation}
    Let $\mathcal{D}_c = [\mu_X - c\,\sigma_X,\, \mu_X + c\,\sigma_X] \times [\mu_Y - c\,\sigma_Y,\, \mu_Y + c\,\sigma_Y]$ for $c > 0$, and let $I_c(v) = \iint_{\mathcal{D}_c} h(x,y;v) \, dy \, dx$ denote the correspondingly truncated integral. Then for every $v \in \mathbb{R}$ and every $\bs{\Sigma} \succ 0$,
    \begin{equation}
        \label{eq:trunc_bound}
        \left| I(v) - I_c(v) \right| \;\leq\; P\big((X,Y) \notin \mathcal{D}_c\big) \;\leq\; 4\,\Phi(-c) .
    \end{equation}
\end{corollary}

\begin{proof}
    Applying the domination bound (\ref{eq:domination}) yields:
    \begin{equation}
        \left| I(v) - I_c(v) \right| = \left| \iint_{\mathbb{R}^2 \setminus \mathcal{D}_c} h(x,y;v) \, dy \, dx \right| \leq \iint_{\mathbb{R}^2 \setminus \mathcal{D}_c} f_{X,Y}(x,y) \, dy \, dx = P\big((X,Y) \notin \mathcal{D}_c\big) .
    \end{equation}
    The complement $\mathbb{R}^2 \setminus \mathcal{D}_c$ is contained in the union of the marginal exceedance events $\{|X - \mu_X| > c\,\sigma_X\}$ and $\{|Y - \mu_Y| > c\,\sigma_Y\}$. Under bivariate normality, each marginal distribution is univariate Gaussian with $P(|X - \mu_X| > c\,\sigma_X) = 2\Phi(-c)$. By Boole's inequality (the union bound), $P((X,Y) \notin \mathcal{D}_c) \leq 2\Phi(-c) + 2\Phi(-c) = 4\Phi(-c)$. \qed
\end{proof}

The bound (\ref{eq:trunc_bound}) is uniform in $v$ and independent of the correlation structure in $\bs{\Sigma}$; Section~\ref{subsec:stability} applies this bound to justify the domain half-width $c = 8$ used in numerical quadrature.

\begin{remark}[Symmetry and Tail Weight of $V$]
\label{rem:moments}
Two exact distributional properties characterize the product $V = XYZ$. First, when $\bs{\mu} = \bs{0}$, the joint distribution of $(X,Y,Z)$ is invariant under the sign reflection $(X,Y,Z) \mapsto (-X,-Y,-Z)$, which maps $V$ to $-V$. The distribution of $V$ is therefore symmetric about zero, and all odd central moments vanish for every $\bs{\Sigma} \succ 0$, including mixed-sign correlation structures. Second, the even moments admit exact closed forms. Let $\rho_{ij}$ denote the pairwise correlations, and define $e_1 = \rho_{12}^2 + \rho_{13}^2 + \rho_{23}^2$, $e_2 = \rho_{12}^2\rho_{13}^2 + \rho_{12}^2\rho_{23}^2 + \rho_{13}^2\rho_{23}^2$, and $p = \rho_{12}\rho_{13}\rho_{23}$. Evaluating the Gaussian product-moment pairings \citep[\S2.6.2]{anderson2003} on the correlation scale yields:
\begin{align*}
    E[V^2] &= 1 + 2e_1 + 8p , \\
    E[V^4] &= 27 + 216 e_1 + 72 e_1^2 + 720 e_2 + 1728 p + 1152 p\, e_1 + 1728 p^2 .
\end{align*}
The excess kurtosis $\gamma_2 = E[V^4] / (E[V^2])^2 - 3$ is scale-invariant and holds for all $\bs{\Sigma} \succ 0$. Under independence ($\bs{\Sigma} = \mathbf{I}_3$), $\gamma_2 = 24$. In the Mixed-Sign Correlation scenario of Section~\ref{sec:simulation} ($\rho_{12} = 0.5, \rho_{13} = -0.3, \rho_{23} = 0.2$), we obtain $E[V^2] = 38/25$, $E[V^4] = 51282/625$, and $\gamma_2 = 23475/722 \approx 32.5$, as reported in Section~\ref{subsec:edge}.
\end{remark}

\section{Degeneracy of the Delta-Method Variance for the Triple Product}
\label{app:delta}

In sequential mediation models ($A \to M_1 \to M_2 \to Y$), the three-path indirect effect is the product of three path coefficients $\theta = \beta_1 \beta_2 \beta_3$. Let $\hat{\bs{\beta}} = (\hat\beta_1, \hat\beta_2, \hat\beta_3)^\T$ denote the coefficient estimators, with asymptotic distribution $\hat{\bs{\beta}} \sim \mathcal{N}_3(\bs{\beta}, \bs{\Sigma}_{\beta}/N)$, where $\bs{\Sigma}_{\beta}$ is the asymptotic covariance of $\sqrt{N}(\hat{\bs{\beta}} - \bs{\beta})$, as in Section~\ref{subsec:coverage}. Setting $(X, Y, Z) \equiv (\hat\beta_1, \hat\beta_2, \hat\beta_3)$ and $\bs{\Sigma} \equiv \bs{\Sigma}_{\beta}/N$, the point estimator $\hat\theta = \hat\beta_1 \hat\beta_2 \hat\beta_3$ coincides with the random product $V = XYZ$. Writing $g(\bs{\beta}) = \beta_1 \beta_2 \beta_3$, the first-order Taylor series (Delta method) approximates the sampling variance as:
\begin{equation}
    \Var(\hat\theta) \approx N^{-1} \nabla g(\bs{\beta})^\T \bs{\Sigma}_{\beta} \, \nabla g(\bs{\beta}) ,
    \qquad
    \nabla g(\bs{\beta}) = (\beta_2 \beta_3, \ \beta_1 \beta_3, \ \beta_1 \beta_2)^\T .
\end{equation}

\begin{lemma}[Gradient Collapse of the Triple Product]
    \label{lem:delta}
    The gradient $\nabla g(\bs{\beta})$ equals $\bs{0}$ if and only if at least two of the three coefficients $\beta_1, \beta_2, \beta_3$ equal zero. Consequently, the first-order asymptotic variance $\nabla g^\T \bs{\Sigma}_{\beta} \nabla g$ vanishes whenever any two coefficients are zero, for every positive-definite covariance matrix $\bs{\Sigma}_{\beta}$.
\end{lemma}

\begin{proof}
    ($\Leftarrow$) If at least two coefficients vanish, say $\beta_i = \beta_j = 0$ for $i \neq j$, each component of $\nabla g(\bs{\beta})$ contains at least one of $\beta_i$ or $\beta_j$ as a multiplicative factor and therefore equals zero, yielding $\nabla g(\bs{\beta}) = \bs{0}$.

    ($\Rightarrow$) Conversely, suppose $\nabla g(\bs{\beta}) = \bs{0}$, so that $\beta_2 \beta_3 = \beta_1 \beta_3 = \beta_1 \beta_2 = 0$. If fewer than two coefficients were zero, at least two coefficients would be non-zero (for instance, $\beta_1 \neq 0$ and $\beta_2 \neq 0$), which implies $\beta_1 \beta_2 \neq 0$ and contradicts the third condition. Hence, at least two coefficients must vanish.

    When $\nabla g(\bs{\beta}) = \bs{0}$, the quadratic form $\nabla g^\T \bs{\Sigma}_{\beta} \nabla g$ evaluates to zero for every covariance matrix $\bs{\Sigma}_{\beta}$, collapsing the first-order variance to zero. \qed
\end{proof}

The two-coefficient threshold is sharp. If only one coefficient approaches zero (for example, $\beta_3 \to 0$ with $\beta_1, \beta_2 \neq 0$), the gradient converges to $(0, 0, \beta_1 \beta_2)^\T \neq \bs{0}$, and the asymptotic variance remains strictly positive: $(\beta_1 \beta_2)^2 \, \Var(\hat\beta_3) > 0$. Variance collapse is therefore specific to configurations with two or more small coefficients. In that parameter regime, the true first-order variance shrinks to zero while the sample estimator $\hat\theta$ remains stochastic, rendering the Wald statistic $\hat\theta / \widehat{\mathrm{se}}(\hat\theta)$ and its confidence intervals poorly calibrated. Section~\ref{subsec:coverage} demonstrates the resulting boundary undercoverage empirically.

\begin{remark}[Null Behavior of the Plug-in Wald Statistic]
\label{rem:plugin}
Lemma~\ref{lem:delta} evaluates the asymptotic variance at the true parameter $\bs{\beta}$. In empirical practice, researchers substitute the sample estimates $\hat{\bs{\beta}}$, for which $\nabla g(\hat{\bs{\beta}}) \neq \bs{0}$ almost surely under continuous sampling distributions. The plug-in variance does not vanish, but its behavior under the complete null $\bs{\beta} = \bs{0}$ reveals extreme conservatism. Under diagonal covariance $\bs{\Sigma}_{\beta}/N = \diag(\sigma_1^2, \sigma_2^2, \sigma_3^2)$, where $\sigma_i$ denotes the standard error of $\hat\beta_i$, the plug-in variance is $\widehat{\mathrm{se}}^2(\hat\theta) = \hat\theta^2 \sum_{i=1}^3 (\sigma_i / \hat\beta_i)^2$. The Wald statistic $T = \hat\theta / \widehat{\mathrm{se}}(\hat\theta)$ satisfies $T^{-2} = \sum_{i=1}^3 (\sigma_i / \hat\beta_i)^2$, which implies the upper bound $|T| \le \min_{i} |\hat\beta_i| / \sigma_i$. Rejecting the null hypothesis at significance level $\alpha$ therefore requires all three coefficients to be individually significant at level $\alpha$.

Because the estimators $\hat\beta_i$ are independent under diagonal covariance, the null rejection rate cannot exceed $\alpha^3$ ($1.25 \times 10^{-4}$ at $\alpha = 0.05$). Under joint normality, this rejection rate admits an exact closed form. Each term $(\sigma_i / \hat\beta_i)^2$ has CDF $2\Phi(-x^{-1/2})$, which is the one-sided L\'evy distribution with stability index $1/2$. The sum of three independent standard L\'evy variables follows a L\'evy distribution scaled by $3^2 = 9$ \citep[Example~1.3 and Eq.~(1.8)]{nolan2020}. The rejection probability $\{T^2 > z_{\alpha/2}^2\}$ is therefore:
\begin{equation}
    P\left(\sum_{i=1}^3 \frac{\sigma_i^2}{\hat\beta_i^2} < z_{\alpha/2}^{-2}\right) = 2\,\Phi(-3\,z_{\alpha/2}) = 4.1 \times 10^{-9}
\end{equation}
at $\alpha = 0.05$, where $z_{\alpha/2} \approx 1.96$ is the standard normal critical value. The plug-in Wald interval thus exhibits severe overcoverage under the sharp null ($\bs{\beta} = \bs{0}$), explaining the zero-rejection rates observed across simulation replications \citep{mackinnon2002}. While this conservatism prevents false positives at the origin, the underlying gradient collapse produces severe undercoverage near parameter boundaries (Section~\ref{subsec:coverage}).
\end{remark}

\bibliographystyle{elsarticle-harv}
\bibliography{references}

\begin{thebibliography}{45}
\expandafter\ifx\csname natexlab\endcsname\relax\def\natexlab#1{#1}\fi
\providecommand{\url}[1]{\texttt{#1}}
\providecommand{\href}[2]{#2}
\providecommand{\path}[1]{#1}
\providecommand{\DOIprefix}{doi:}
\providecommand{\ArXivprefix}{arXiv:}
\providecommand{\URLprefix}{URL: }
\providecommand{\Pubmedprefix}{pmid:}
\providecommand{\doi}[1]{\href{http://dx.doi.org/#1}{\path{#1}}}
\providecommand{\Pubmed}[1]{\href{pmid:#1}{\path{#1}}}
\providecommand{\bibinfo}[2]{#2}
\ifx\xfnm\relax \def\xfnm[#1]{\unskip,\space#1}\fi
\bibitem[{Anderson(2003)}]{anderson2003}
\bibinfo{author}{Anderson, T.W.}, \bibinfo{year}{2003}.
\newblock \bibinfo{title}{An Introduction to Multivariate Statistical
  Analysis}.
\newblock \bibinfo{edition}{3rd} ed., \bibinfo{publisher}{Wiley},
  \bibinfo{address}{Hoboken, NJ}.
\bibitem[{Andrews(2000)}]{andrews2000}
\bibinfo{author}{Andrews, D.W.K.}, \bibinfo{year}{2000}.
\newblock \bibinfo{title}{Inconsistency of the bootstrap when a parameter is on
  the boundary of the parameter space}.
\newblock \bibinfo{journal}{Econometrica} \bibinfo{volume}{68},
  \bibinfo{pages}{399--405}.
\newblock \DOIprefix\doi{10.1111/1468-0262.00114}.
\bibitem[{Axler(2020)}]{axler2020}
\bibinfo{author}{Axler, S.}, \bibinfo{year}{2020}.
\newblock \bibinfo{title}{Measure, Integration \& Real Analysis}. volume
  \bibinfo{volume}{282} of \textit{\bibinfo{series}{Graduate Texts in
  Mathematics}}.
\newblock \bibinfo{publisher}{Springer International Publishing},
  \bibinfo{address}{Cham}.
\newblock \DOIprefix\doi{10.1007/978-3-030-33143-6}.
\bibitem[{Bakhvalov(1959)}]{bakhvalov1959}
\bibinfo{author}{Bakhvalov, N.S.}, \bibinfo{year}{1959}.
\newblock \bibinfo{title}{On the approximate calculation of multiple
  integrals}.
\newblock \bibinfo{journal}{Vestnik Moskovskogo Universiteta, Seriya
  Matematiki, Mehaniki, Astronomii, Fiziki, Himii} \bibinfo{volume}{4},
  \bibinfo{pages}{3--18}.
\newblock \bibinfo{note}{In Russian. English translation: Journal of Complexity
  31 (2015) 502--516, doi:10.1016/j.jco.2014.12.003}.
\bibitem[{Berntsen et~al.(1991)Berntsen, Espelid and Genz}]{berntsen1991}
\bibinfo{author}{Berntsen, J.}, \bibinfo{author}{Espelid, T.O.},
  \bibinfo{author}{Genz, A.}, \bibinfo{year}{1991}.
\newblock \bibinfo{title}{An adaptive algorithm for the approximate calculation
  of multiple integrals}.
\newblock \bibinfo{journal}{ACM Transactions on Mathematical Software}
  \bibinfo{volume}{17}, \bibinfo{pages}{437--451}.
\newblock \DOIprefix\doi{10.1145/210232.210233}.
\bibitem[{Bickel and Freedman(1981)}]{bickel1981}
\bibinfo{author}{Bickel, P.J.}, \bibinfo{author}{Freedman, D.A.},
  \bibinfo{year}{1981}.
\newblock \bibinfo{title}{Some asymptotic theory for the bootstrap}.
\newblock \bibinfo{journal}{Annals of Statistics} \bibinfo{volume}{9},
  \bibinfo{pages}{1196--1217}.
\newblock \DOIprefix\doi{10.1214/aos/1176345637}.
\bibitem[{Bohrnstedt and Goldberger(1969)}]{bohrnstedt1969}
\bibinfo{author}{Bohrnstedt, G.W.}, \bibinfo{author}{Goldberger, A.S.},
  \bibinfo{year}{1969}.
\newblock \bibinfo{title}{On the exact covariance of products of random
  variables}.
\newblock \bibinfo{journal}{Journal of the American Statistical Association}
  \bibinfo{volume}{64}, \bibinfo{pages}{1439--1442}.
\newblock \DOIprefix\doi{10.2307/2286081},
  \href{http://arxiv.org/abs/2286081}{{\tt arXiv:2286081}}.
\bibitem[{Craig(1936)}]{craig1936}
\bibinfo{author}{Craig, C.C.}, \bibinfo{year}{1936}.
\newblock \bibinfo{title}{On the frequency function of xy}.
\newblock \bibinfo{journal}{The Annals of Mathematical Statistics}
  \bibinfo{volume}{7}, \bibinfo{pages}{1--15}.
\newblock \DOIprefix\doi{10.1214/aoms/1177732541}.
\bibitem[{Cui et~al.(2016)Cui, Yu, Iommelli and Kong}]{cui2016}
\bibinfo{author}{Cui, G.}, \bibinfo{author}{Yu, X.}, \bibinfo{author}{Iommelli,
  S.}, \bibinfo{author}{Kong, L.}, \bibinfo{year}{2016}.
\newblock \bibinfo{title}{Exact distribution for the product of two correlated
  gaussian random variables}.
\newblock \bibinfo{journal}{IEEE Signal Processing Letters}
  \bibinfo{volume}{23}, \bibinfo{pages}{1662--1666}.
\newblock \DOIprefix\doi{10.1109/LSP.2016.2614539}.
\bibitem[{Davis and Rabinowitz(1984)}]{davis1984}
\bibinfo{author}{Davis, P.J.}, \bibinfo{author}{Rabinowitz, P.},
  \bibinfo{year}{1984}.
\newblock \bibinfo{title}{Methods of Numerical Integration}.
\newblock \bibinfo{edition}{2nd} ed., \bibinfo{publisher}{Academic Press}.
\bibitem[{Drezner(1994)}]{drezner1994}
\bibinfo{author}{Drezner, Z.}, \bibinfo{year}{1994}.
\newblock \bibinfo{title}{Computation of the trivariate normal integral}.
\newblock \bibinfo{journal}{Mathematics of Computation} \bibinfo{volume}{62},
  \bibinfo{pages}{289--294}.
\newblock \DOIprefix\doi{10.2307/2153409},
  \href{http://arxiv.org/abs/2153409}{{\tt arXiv:2153409}}.
\bibitem[{Gaunt and Ye(2025)}]{gaunt2025}
\bibinfo{author}{Gaunt, R.E.}, \bibinfo{author}{Ye, Z.}, \bibinfo{year}{2025}.
\newblock \bibinfo{title}{Asymptotic approximations for the distribution of the
  product of correlated normal random variables}.
\newblock \bibinfo{journal}{Journal of Mathematical Analysis and Applications}
  \bibinfo{volume}{543}, \bibinfo{pages}{128987}.
\newblock \DOIprefix\doi{10.1016/j.jmaa.2024.128987}.
\bibitem[{Genz and Bretz(2009)}]{genz2009}
\bibinfo{author}{Genz, A.}, \bibinfo{author}{Bretz, F.}, \bibinfo{year}{2009}.
\newblock \bibinfo{title}{Computation of Multivariate Normal and t
  Probabilities}.
\newblock Lecture Notes in Statistics, \bibinfo{publisher}{Springer-Verlag},
  \bibinfo{address}{Heidelberg}.
\newblock \DOIprefix\doi{10.1007/978-3-642-01689-9}.
\bibitem[{Genz and Malik(1980)}]{genzmalik1980}
\bibinfo{author}{Genz, A.C.}, \bibinfo{author}{Malik, A.A.},
  \bibinfo{year}{1980}.
\newblock \bibinfo{title}{Remarks on algorithm 006: An adaptive algorithm for
  numerical integration over an {N}-dimensional rectangular region}.
\newblock \bibinfo{journal}{Journal of Computational and Applied Mathematics}
  \bibinfo{volume}{6}, \bibinfo{pages}{295--302}.
\newblock \DOIprefix\doi{10.1016/0771-050X(80)90039-X}.
\bibitem[{Gneiting and Raftery(2007)}]{gneiting2007}
\bibinfo{author}{Gneiting, T.}, \bibinfo{author}{Raftery, A.E.},
  \bibinfo{year}{2007}.
\newblock \bibinfo{title}{Strictly proper scoring rules, prediction, and
  estimation}.
\newblock \bibinfo{journal}{Journal of the American Statistical Association}
  \bibinfo{volume}{102}, \bibinfo{pages}{359--378}.
\newblock \DOIprefix\doi{10.1198/016214506000001437}.
\bibitem[{Gonnet(2012)}]{gonnet2012}
\bibinfo{author}{Gonnet, P.}, \bibinfo{year}{2012}.
\newblock \bibinfo{title}{A review of error estimation in adaptive quadrature}.
\newblock \bibinfo{journal}{ACM Computing Surveys} \bibinfo{volume}{44},
  \bibinfo{pages}{1--36}.
\newblock \DOIprefix\doi{10.1145/2333112.2333117}.
\bibitem[{Goodman(1960)}]{goodman1960}
\bibinfo{author}{Goodman, L.A.}, \bibinfo{year}{1960}.
\newblock \bibinfo{title}{On the exact variance of products}.
\newblock \bibinfo{journal}{Journal of the American Statistical Association}
  \bibinfo{volume}{55}, \bibinfo{pages}{708--713}.
\newblock \DOIprefix\doi{10.2307/2281592},
  \href{http://arxiv.org/abs/2281592}{{\tt arXiv:2281592}}.
\bibitem[{Goodman(1962)}]{goodman1962}
\bibinfo{author}{Goodman, L.A.}, \bibinfo{year}{1962}.
\newblock \bibinfo{title}{The variance of the product of {{K}} random
  variables}.
\newblock \bibinfo{journal}{Journal of the American Statistical Association}
  \bibinfo{volume}{57}, \bibinfo{pages}{54--60}.
\newblock \DOIprefix\doi{10.2307/2282440},
  \href{http://arxiv.org/abs/2282440}{{\tt arXiv:2282440}}.
\bibitem[{Higham(2002)}]{higham2002}
\bibinfo{author}{Higham, N.J.}, \bibinfo{year}{2002}.
\newblock \bibinfo{title}{Accuracy and Stability of Numerical Algorithms}.
\newblock \bibinfo{edition}{2nd} ed., \bibinfo{publisher}{Society for
  Industrial and Applied Mathematics}, \bibinfo{address}{Philadelphia, PA}.
\newblock \DOIprefix\doi{10.1137/1.9780898718027}.
\bibitem[{Lomnicki(1967)}]{lomnicki1967}
\bibinfo{author}{Lomnicki, Z.A.}, \bibinfo{year}{1967}.
\newblock \bibinfo{title}{On the distribution of products of random variables}.
\newblock \bibinfo{journal}{Journal of the Royal Statistical Society. Series B
  (Methodological)} \bibinfo{volume}{29}, \bibinfo{pages}{513--524}.
\newblock \DOIprefix\doi{10.1111/j.2517-6161.1967.tb00713.x}.
\bibitem[{Lyness and Kaganove(1976)}]{lyness1976}
\bibinfo{author}{Lyness, J.N.}, \bibinfo{author}{Kaganove, J.J.},
  \bibinfo{year}{1976}.
\newblock \bibinfo{title}{Comments on the nature of automatic quadrature
  routines}.
\newblock \bibinfo{journal}{ACM Transactions on Mathematical Software}
  \bibinfo{volume}{2}, \bibinfo{pages}{65--81}.
\newblock \DOIprefix\doi{10.1145/355666.355671}.
\bibitem[{MacKinnon(2008)}]{mackinnon2008}
\bibinfo{author}{MacKinnon, D.P.}, \bibinfo{year}{2008}.
\newblock \bibinfo{title}{Introduction to Statistical Mediation Analysis}.
\newblock \bibinfo{publisher}{Lawrence Erlbaum Associates},
  \bibinfo{address}{New York}.
\newblock \DOIprefix\doi{10.4324/9780203809556}.
\bibitem[{MacKinnon et~al.(2002)MacKinnon, Lockwood, Hoffman, West and
  Sheets}]{mackinnon2002}
\bibinfo{author}{MacKinnon, D.P.}, \bibinfo{author}{Lockwood, C.M.},
  \bibinfo{author}{Hoffman, J.M.}, \bibinfo{author}{West, S.G.},
  \bibinfo{author}{Sheets, V.}, \bibinfo{year}{2002}.
\newblock \bibinfo{title}{A comparison of methods to test mediation and other
  intervening variable effects}.
\newblock \bibinfo{journal}{Psychological Methods} \bibinfo{volume}{7},
  \bibinfo{pages}{83--104}.
\newblock \DOIprefix\doi{10.1037/1082-989X.7.1.83}.
\bibitem[{MacKinnon et~al.(2004)MacKinnon, Lockwood and
  Williams}]{mackinnon2004}
\bibinfo{author}{MacKinnon, D.P.}, \bibinfo{author}{Lockwood, C.M.},
  \bibinfo{author}{Williams, J.}, \bibinfo{year}{2004}.
\newblock \bibinfo{title}{Confidence limits for the indirect effect:
  Distribution of the product and resampling methods}.
\newblock \bibinfo{journal}{Multivariate Behavioral Research}
  \bibinfo{volume}{39}, \bibinfo{pages}{99--128}.
\newblock \DOIprefix\doi{10.1207/s15327906mbr3901_4}.
\bibitem[{Meeker and Escobar(1994)}]{meeker1994}
\bibinfo{author}{Meeker, W.Q.}, \bibinfo{author}{Escobar, L.A.},
  \bibinfo{year}{1994}.
\newblock \bibinfo{title}{An algorithm to compute the cdf of the product of two
  normal random variables}.
\newblock \bibinfo{journal}{Communications in Statistics - Simulation and
  Computation} \bibinfo{volume}{23}, \bibinfo{pages}{271--280}.
\newblock \DOIprefix\doi{10.1080/03610919408813168}.
\bibitem[{Morris et~al.(2019)Morris, White and Crowther}]{morris2019}
\bibinfo{author}{Morris, T.P.}, \bibinfo{author}{White, I.R.},
  \bibinfo{author}{Crowther, M.J.}, \bibinfo{year}{2019}.
\newblock \bibinfo{title}{Using simulation studies to evaluate statistical
  methods}.
\newblock \bibinfo{journal}{Statistics in Medicine} \bibinfo{volume}{38},
  \bibinfo{pages}{2074--2102}.
\newblock \DOIprefix\doi{10.1002/sim.8086}.
\bibitem[{Narasimhan et~al.(2025)Narasimhan, Johnson, Hahn, Bouvier and
  Ki{\^e}u}]{cubature}
\bibinfo{author}{Narasimhan, B.}, \bibinfo{author}{Johnson, S.G.},
  \bibinfo{author}{Hahn, T.}, \bibinfo{author}{Bouvier, A.},
  \bibinfo{author}{Ki{\^e}u, K.}, \bibinfo{year}{2025}.
\newblock \bibinfo{title}{cubature: Adaptive Multivariate Integration over
  Hypercubes}.
\newblock \DOIprefix\doi{10.32614/CRAN.package.cubature}. \bibinfo{note}{r
  package version 2.1.4-1, published 2025-11-22}.
\bibitem[{Nolan(2020)}]{nolan2020}
\bibinfo{author}{Nolan, J.P.}, \bibinfo{year}{2020}.
\newblock \bibinfo{title}{Univariate Stable Distributions: Models for Heavy
  Tailed Data}.
\newblock Springer Series in Operations Research and Financial Engineering,
  \bibinfo{publisher}{Springer International Publishing},
  \bibinfo{address}{Cham}.
\newblock \DOIprefix\doi{10.1007/978-3-030-52915-4}.
\bibitem[{Novak(2016)}]{novak2016}
\bibinfo{author}{Novak, E.}, \bibinfo{year}{2016}.
\newblock \bibinfo{title}{Some results on the complexity of numerical
  integration}, in: \bibinfo{editor}{Cools, R.}, \bibinfo{editor}{Nuyens, D.}
  (Eds.), \bibinfo{booktitle}{Monte Carlo and Quasi-Monte Carlo Methods:
  {MCQMC}, {L}euven, {B}elgium, {A}pril 2014}. \bibinfo{publisher}{Springer
  International Publishing}. volume \bibinfo{volume}{163} of
  \textit{\bibinfo{series}{Springer Proceedings in Mathematics \& Statistics}},
  pp. \bibinfo{pages}{161--183}.
\newblock \DOIprefix\doi{10.1007/978-3-319-33507-0_6}.
\bibitem[{Ogasawara(2025)}]{ogasawara2025}
\bibinfo{author}{Ogasawara, H.}, \bibinfo{year}{2025}.
\newblock \bibinfo{title}{The product distribution of more-than-two correlated
  normal variables: {A} didactic review with some new findings}.
\newblock \bibinfo{journal}{Communications in Statistics---Theory and Methods}
  \bibinfo{volume}{54}, \bibinfo{pages}{5354--5377}.
\newblock \DOIprefix\doi{10.1080/03610926.2024.2437487}.
\bibitem[{Pearl(2001)}]{pearl2001}
\bibinfo{author}{Pearl, J.}, \bibinfo{year}{2001}.
\newblock \bibinfo{title}{Direct and indirect effects}, in:
  \bibinfo{booktitle}{Proceedings of the Seventeenth Conference on Uncertainty
  in Artificial Intelligence (UAI)}, \bibinfo{publisher}{Morgan Kaufmann}. pp.
  \bibinfo{pages}{411--420}.
\bibitem[{Pearson et~al.(1929)Pearson, Jeffery and Elderton}]{pearson1929}
\bibinfo{author}{Pearson, K.}, \bibinfo{author}{Jeffery, G.B.},
  \bibinfo{author}{Elderton, E.M.}, \bibinfo{year}{1929}.
\newblock \bibinfo{title}{On the distribution of the first product
  moment-coefficient, in samples drawn from an indefinitely large normal
  population}.
\newblock \bibinfo{journal}{Biometrika} \bibinfo{volume}{21},
  \bibinfo{pages}{164--193}.
\newblock \DOIprefix\doi{10.1093/biomet/21.1-4.164}.
\bibitem[{Preacher and Hayes(2008)}]{preacher2008}
\bibinfo{author}{Preacher, K.J.}, \bibinfo{author}{Hayes, A.F.},
  \bibinfo{year}{2008}.
\newblock \bibinfo{title}{Asymptotic and resampling strategies for assessing
  and comparing indirect effects in multiple mediator models}.
\newblock \bibinfo{journal}{Behavior Research Methods} \bibinfo{volume}{40},
  \bibinfo{pages}{879--891}.
\newblock \DOIprefix\doi{10.3758/BRM.40.3.879}.
\bibitem[{Preacher and Kelley(2011)}]{preacher2011a}
\bibinfo{author}{Preacher, K.J.}, \bibinfo{author}{Kelley, K.},
  \bibinfo{year}{2011}.
\newblock \bibinfo{title}{Effect size measures for mediation models:
  Quantitative strategies for communicating indirect effects}.
\newblock \bibinfo{journal}{Psychological Methods} \bibinfo{volume}{16},
  \bibinfo{pages}{93--115}.
\newblock \DOIprefix\doi{10.1037/a0022658}.
\bibitem[{Robins and Greenland(1992)}]{robins1992}
\bibinfo{author}{Robins, J.M.}, \bibinfo{author}{Greenland, S.},
  \bibinfo{year}{1992}.
\newblock \bibinfo{title}{Identifiability and exchangeability for direct and
  indirect effects}.
\newblock \bibinfo{journal}{Epidemiology} \bibinfo{volume}{3},
  \bibinfo{pages}{143--155}.
\newblock \DOIprefix\doi{10.1097/00001648-199203000-00013}.
\bibitem[{Springer and Thompson(1966)}]{springer1966}
\bibinfo{author}{Springer, M.D.}, \bibinfo{author}{Thompson, W.E.},
  \bibinfo{year}{1966}.
\newblock \bibinfo{title}{The distribution of products of independent random
  variables}.
\newblock \bibinfo{journal}{SIAM Journal on Applied Mathematics}
  \bibinfo{volume}{14}, \bibinfo{pages}{511--526}.
\newblock \DOIprefix\doi{10.1137/0114046}.
\bibitem[{Taylor et~al.(2008)Taylor, MacKinnon and Tein}]{taylor2008}
\bibinfo{author}{Taylor, A.B.}, \bibinfo{author}{MacKinnon, D.P.},
  \bibinfo{author}{Tein, J.Y.}, \bibinfo{year}{2008}.
\newblock \bibinfo{title}{Tests of the three-path mediated effect}.
\newblock \bibinfo{journal}{Organizational Research Methods}
  \bibinfo{volume}{11}, \bibinfo{pages}{241--269}.
\newblock \DOIprefix\doi{10.1177/1094428107300344}.
\bibitem[{Tofighi and Kelley(2020)}]{tofighi2020}
\bibinfo{author}{Tofighi, D.}, \bibinfo{author}{Kelley, K.},
  \bibinfo{year}{2020}.
\newblock \bibinfo{title}{Indirect effects in sequential mediation models:
  Evaluating methods for hypothesis testing and confidence interval formation}.
\newblock \bibinfo{journal}{Multivariate Behavioral Research}
  \bibinfo{volume}{55}, \bibinfo{pages}{188--210}.
\newblock \DOIprefix\doi{10.1080/00273171.2019.1618545}.
\bibitem[{Tofighi and MacKinnon(2011)}]{tofighi2011}
\bibinfo{author}{Tofighi, D.}, \bibinfo{author}{MacKinnon, D.P.},
  \bibinfo{year}{2011}.
\newblock \bibinfo{title}{Rmediation: An {R} package for mediation analysis
  confidence intervals}.
\newblock \bibinfo{journal}{Behavior Research Methods} \bibinfo{volume}{43},
  \bibinfo{pages}{692--700}.
\newblock \DOIprefix\doi{10.3758/s13428-011-0076-x}.
\bibitem[{Valeri and VanderWeele(2013)}]{valeri2013}
\bibinfo{author}{Valeri, L.}, \bibinfo{author}{VanderWeele, T.J.},
  \bibinfo{year}{2013}.
\newblock \bibinfo{title}{Mediation analysis allowing for exposure--mediator
  interactions and causal interpretation: Theoretical assumptions and
  implementation with {SAS} and {SPSS} macros}.
\newblock \bibinfo{journal}{Psychological Methods} \bibinfo{volume}{18},
  \bibinfo{pages}{137--150}.
\newblock \DOIprefix\doi{10.1037/a0031034}.
\bibitem[{VanderWeele(2015)}]{vanderweele2015}
\bibinfo{author}{VanderWeele, T.J.}, \bibinfo{year}{2015}.
\newblock \bibinfo{title}{Explanation in Causal Inference: Methods for
  Mediation and Interaction}.
\newblock \bibinfo{publisher}{Oxford University Press}, \bibinfo{address}{New
  York, NY}.
\bibitem[{VanderWeele and Vansteelandt(2014)}]{vanderweele2014c}
\bibinfo{author}{VanderWeele, T.J.}, \bibinfo{author}{Vansteelandt, S.},
  \bibinfo{year}{2014}.
\newblock \bibinfo{title}{Mediation analysis with multiple mediators}.
\newblock \bibinfo{journal}{Epidemiologic Methods} \bibinfo{volume}{2},
  \bibinfo{pages}{95--115}.
\newblock \DOIprefix\doi{10.1515/em-2012-0010}.
\bibitem[{Vinokur and Schul(1997)}]{vinokur1997}
\bibinfo{author}{Vinokur, A.D.}, \bibinfo{author}{Schul, Y.},
  \bibinfo{year}{1997}.
\newblock \bibinfo{title}{Mastery and inoculation against setbacks as active
  ingredients in the {JOBS} intervention for the unemployed}.
\newblock \bibinfo{journal}{Journal of Consulting and Clinical Psychology}
  \bibinfo{volume}{65}, \bibinfo{pages}{867--877}.
\newblock \DOIprefix\doi{10.1037/0022-006X.65.5.867}.
\bibitem[{Wishart and Bartlett(1932)}]{wishart1932}
\bibinfo{author}{Wishart, J.}, \bibinfo{author}{Bartlett, M.S.},
  \bibinfo{year}{1932}.
\newblock \bibinfo{title}{The distribution of second order moment statistics in
  a normal system}.
\newblock \bibinfo{journal}{Mathematical Proceedings of the Cambridge
  Philosophical Society} \bibinfo{volume}{28}, \bibinfo{pages}{455--459}.
\newblock \DOIprefix\doi{10.1017/S0305004100010690}.
\bibitem[{Zheng et~al.(2012)Zheng, Wei, Hamalainen and Tirkkonen}]{zheng2012}
\bibinfo{author}{Zheng, Z.}, \bibinfo{author}{Wei, L.},
  \bibinfo{author}{Hamalainen, J.}, \bibinfo{author}{Tirkkonen, O.},
  \bibinfo{year}{2012}.
\newblock \bibinfo{title}{Approximation to distribution of product of random
  variables using orthogonal polynomials for lognormal density}.
\newblock \bibinfo{journal}{IEEE Communications Letters} \bibinfo{volume}{16},
  \bibinfo{pages}{2028--2031}.
\newblock \DOIprefix\doi{10.1109/lcomm.2012.101712.122141}.

\end{thebibliography}

\end{document}